\documentclass[11pt,reqno]{article}
\usepackage{jmlr2e}
\usepackage{fullpage,times,graphicx,amssymb,amsmath,amsfonts,bbm,psfrag}
\usepackage{booktabs}
\usepackage[table]{xcolor}
\usepackage{epstopdf}
\usepackage{commath}
\usepackage{thmtools}
\usepackage{thm-restate}
\usepackage{natbib}
\usepackage{dirtytalk}
\usepackage{empheq}
\usepackage{subfig,float}
\usepackage{pifont}
\usepackage{tabularx}
\usepackage{array}
\usepackage{mathtools}
\usepackage{algorithm,algcompatible,algpseudocode}
\usepackage{multicol,lipsum}

\usepackage{multirow}
\usepackage{siunitx}
\usepackage[table]{xcolor}
\usepackage{etoolbox}

\robustify\bfseries

\definecolor{referencegray}{gray}{0.93}
\definecolor{methodgray}{gray}{0.97}

\usepackage[most]{tcolorbox}
\newtcolorbox[auto counter, number within=section]{note}[2][]{colback=yellow!10!white,colframe=yellow!50!black,fonttitle=\bfseries, title=Note~\thetcbcounter: #2,#1}
\usepackage{caption}
\definecolor{ddarkbrown}{rgb}{0.5,0.2,0.05}
\definecolor{bbluegray}{rgb}{0.05,0,0.5}
\usepackage[colorlinks,citecolor=bbluegray,linkcolor=ddarkbrown,urlcolor=blue,breaklinks]{hyperref}

\algnewcommand{\Inputs}[1]{\State \textbf{Inputs: \:}{#1}}
\algnewcommand{\Output}[1]{\State \textbf{Output: \:}{#1}}
\algnewcommand{\Initialize}[1]{\State \textbf{Initialize: \:}{#1}}
\algnewcommand{\IIf}[1]{\State\algorithmicif\ #1\ \algorithmicthen}
\algnewcommand{\EndIIf}{\unskip\ \algorithmicend\ \algorithmicif}

\let \oldsection \section
\renewcommand{\section}{\vspace{3ex plus 1ex}\oldsection}

\newcommand{\BEAS}{\begin{eqnarray*}}
\newcommand{\EEAS}{\end{eqnarray*}}
\newcommand{\BEA}{\begin{eqnarray}}
\newcommand{\EEA}{\end{eqnarray}}

\newcommand{\BEQ}{\begin{equation}}
\newcommand{\EEQ}{\end{equation}}
\newcommand{\BIT}{\begin{itemize}}
\newcommand{\EIT}{\end{itemize}}
\newcommand{\BNUM}{\begin{enumerate}}
\newcommand{\ENUM}{\end{enumerate}}

\newcommand{\BA}{\begin{array}}
\newcommand{\EA}{\end{array}}

\newtheorem{mythm}{Theorem}[section]

\newtheorem{myprop}{Proposition}[section]
\newtheorem{mylem}{Lemma}[section]
\newtheorem{mycor}{Corollary}[section]

\newtheorem{myassump}{Assumption}[section]

\title{Domain Adaptation of Mismatched Proximal Denoiser for Plug-and-Play Image Reconstruction}

\begin{document}

    \author{\name Guixian Xu  \email gxx422@student.bham.ac.uk\\
        \addr School of Mathematics,\\ University of Birmingham \\ \\
         \name Jinglai Li  \email j.li.10@bham.ac.uk\\
        \addr School of Mathematics,\\ University of Birmingham \\ \\
    \name Junqi Tang  \email j.tang.2@bham.ac.uk\\
        \addr School of Mathematics,\\ University of Birmingham \\
        \\
        }

    \editor{}

    \maketitle

\begin{abstract}
Plug-and-play proximal gradient descent (PnP-PGD) enables flexible image reconstruction by using denoisers as implicit priors. In practice, these denoisers are often deployed outside their training domains. Existing analyses establish convergence under structural assumptions on the deployed denoiser, such as requiring it to be a proximal map or a contraction. However, they do not measure how domain mismatch affects convergence of PnP-PGD. We define this effect as \emph{proximal mismatch}: the discrepancy between a deployed denoiser $\widehat{\mathsf D}$ and a target-domain reference map $\mathsf D_\star=\operatorname{prox}_{R_\star}$ associated with the underlying regularizer $R_\star$. Under this mismatch, each denoising update becomes an inexact proximal step for the target objective. We further derive a stationarity bound that decays at a rate of $\mathcal{O}(1/K)$, with an additive term proportional to the average squared proximal mismatch. This result motivates adaptation via proximal matching rather than MSE-based adaptation alone. We study this approach with two established denoiser families: learned proximal networks and gradient-step denoisers. Experiments on Gaussian deblurring and super-resolution under substantial domain shift show that proximal matching adaptation improves reconstruction quality significantly over MSE-based adaptation, yielding the largest numerical gains in the few-shot regime.
\end{abstract}

\section{Introduction}

Imaging inverse problems recover an unknown image from noisy measurements by combining a forward physical model with a prior. Plug-and-play (PnP) methods replace the regularization step of an iterative optimization algorithm with a trained denoiser~\citep{venkatakrishnan2013plug,chan2017plug,kamilov2023plug}. This design separates the prior from the forward model and allows the same denoiser to be used with different imaging operators. However, a generic denoiser may not be the proximal map of any regularizer, so the resulting iterations may not minimize an explicit objective~\citep{buzzard2018plug,kamilov2023plug}. Existing convergence analyses have therefore studied PnP as a fixed-point iteration under conditions such as nonexpansiveness and contractivity~\citep{chan2017plug,ryu2019plug}. More recent work has designed denoisers with proximal structure, including convolutional proximal networks, gradient-step denoisers, and learned proximal networks~\citep{hertrich2021convolutional,hurault2022gradient,hurault2022proximal,fang2024whats}. Under suitable conditions, these denoisers allow PnP-PGD to be analyzed as proximal gradient descent on an explicit objective and provide first-order convergence guarantees.

\noindent
Despite this progress, PnP remains sensitive to prior mismatch. A denoiser trained on one image distribution may be deployed on another, creating a mismatch between the learned and target priors. This mismatch is difficult to correct in few-shot scenarios with limited target data. Even if the deployed denoiser satisfies the conditions required for convergence, the regularizer it represents may not match the target domain. Existing work has analyzed prior mismatch in SD-RED and PnP-ADMM and has explored test-time adaptation for PnP~\citep{shoushtari2022deep,chandler2023overcoming,shoushtari2024prior}. Related studies show that denoiser inputs along PnP trajectories can differ from standard Gaussian denoising samples~\citep{vo2024plug}, and that proximal or MAP behavior differs from MMSE denoising~\citep{Fermanian2023PnPReG,vert2026beyond}. Thus, denoising MSE alone may not target the operator behavior required by PnP reconstruction. In the structured proximal setting, this behavior is represented by a target-domain proximal map. What remains unclear is how replacing this map with a mismatched denoiser affects convergence of the target objective in PnP-PGD and how this effect should guide adaptation. This issue matters because MSE-based denoiser training targets a conditional mean but does not directly enforce agreement with a target proximal map. Therefore, better denoising performance does not necessarily reduce the proximal error encountered during PnP-PGD.

\noindent
We analyze PnP-PGD under domain shift by comparing the deployed denoiser $\widehat{\mathsf D}$ with a target-domain reference proximal map $\mathsf D_\star = \operatorname{prox}_{R_\star}$. This reference map corresponds to the target regularizer $R_\star$ and defines the target objective $F_\star = \eta f + R_\star$. Replacing $\mathsf D_\star$ with $\widehat{\mathsf D}$ turns each denoising update into an inexact proximal step for $F_\star$. We define \emph{proximal mismatch} as the discrepancy between their outputs at the query points generated by PnP-PGD. We then derive a first-order stationarity bound for PnP-PGD. This bound decays at a rate of $\mathcal{O}(1/K)$ and includes an additive term proportional to the average squared proximal mismatch. This result links domain shift directly to convergence of the target objective. Table~\ref{tab:compare-strategy} positions our analysis relative to the most closely related work on inexact proximal methods, prior mismatch, and structured PnP denoisers.

\begin{table}[t]
\centering
\small
\setlength{\tabcolsep}{4pt}
\renewcommand{\arraystretch}{1.12}
\begin{tabular}{
    p{0.22\linewidth}
    p{0.12\linewidth}
    p{0.24\linewidth}
    p{0.26\linewidth}
    p{0.10\linewidth}
}
\toprule
\textbf{Work} &
\textbf{Algorithm} &
\textbf{Prior / denoiser model} &
\textbf{Mismatch / error model} &
\textbf{Assumptions} \\
\midrule
\multicolumn{5}{l}{\emph{Classical / RED-type analyses}} \\
\midrule
\cite{schmidt2011convergence}
& PGD
& Explicit proximal operator
& Inexact proximal steps
& $\epsilon$-optimal \\

\cite{shoushtari2022deep}
& RED
& MMSE denoiser
& Prior mismatch
& BD, BI \\
\midrule
\multicolumn{5}{l}{\emph{Plug-and-Play analyses}} \\
\midrule
\cite{xu2020provable}
& PnP-PGD
& MMSE denoiser
& Prior match
& TD \\

\cite{hurault2022proximal}
& PnP-PGD
& Proximal denoiser*
& Prior match
& TD \\

\cite{shoushtari2024prior}
& PnP-ADMM
& MMSE denoiser
& Prior mismatch
& BD, BI \\

\cite{fang2024whats}
& PnP-PGD
& Proximal denoiser
& Prior match
& TD \\

\rowcolor{gray!12}
Ours
& PnP-PGD
& Proximal denoiser
& Proximal mismatch
& ER \\
\bottomrule
\end{tabular}
\caption{
Theoretical context for the convergence analysis developed in this paper. The table compares representative proximal-gradient, RED, and PnP analyses in terms of algorithm, denoiser model, mismatch/error model, and assumptions. TD = matched target denoiser used at inference; BD = bounded denoiser; BI = bounded iterates; ER = target-envelope regularity. '*' indicates a proximal interpretation that holds under additional conditions.
}
\label{tab:compare-strategy}
\end{table}

\noindent
This bound suggests a direct adaptation strategy. Instead of fine-tuning the denoiser with MSE alone, we use proximal matching to align it with the target-domain proximal map~\citep{fang2024whats}. We study this strategy with two structured denoiser families: learned proximal networks and gradient-step denoisers. This allows us to test whether the same adaptation principle works across two different ways of realizing proximal maps. We evaluate our method on Gaussian deblurring and super-resolution under substantial domain shifts. Proximal-matching adaptation reduces the proximal mismatch along the PnP-PGD trajectory and improves reconstruction quality over MSE-based adaptation. In severe mismatch settings, it improves reconstruction PSNR by more than $10$ dB, with the largest gains in the few-shot regime.

\noindent
Our contributions are threefold:
\begin{itemize}
    \item \textbf{Proximal-mismatch analysis.} We formulate domain shift in PnP-PGD as \emph{proximal mismatch}. We show that this mismatch induces an inexact proximal step and derive an $\mathcal{O}(1/K)$ stationarity bound with an additive term proportional to the average squared proximal mismatch.

    \item \textbf{Proximal-matching adaptation.} Guided by this bound, we propose a few-shot adaptation method based on proximal matching rather than MSE-based adaptation. We apply it to learned proximal networks and gradient-step denoisers.

    \item \textbf{Experimental validation.} We evaluate the proposed method on Gaussian deblurring and super-resolution under substantial domain shifts. Proximal matching adaptation reduces the proximal mismatch along the PnP-PGD trajectory and improves reconstruction PSNR over MSE-based adaptation, with gains exceeding $10$ dB in severe settings and the largest gains in the few-shot regime.

\end{itemize}

\section{Background and Problem Setup}
\label{sec:background}
\paragraph{Inverse Problems.}
We consider the recovery of an unknown signal
$\boldsymbol{x}\in\mathbb{R}^n$ from a noisy measurement
\[
    \boldsymbol{y}
    =
    A\boldsymbol{x}
    +
    \boldsymbol{e},
\]
where $A\in\mathbb{R}^{m\times n}$ is a known forward operator and
$\boldsymbol{e}\in\mathbb{R}^m$ denotes measurement noise. The inverse problem is often ill-posed: a solution may not exist, uniqueness may fail, or small perturbations in the measurements may cause large changes in the recovered signal~\citep{engl1996regularization,hansen2010discrete}. Regularization addresses this instability by adding prior information about the unknown signal. Given a regularizer $R$, we consider
\begin{equation}
    \min_{\boldsymbol{x}\in\mathbb{R}^n}
    F(\boldsymbol{x})
    :=
    \eta f(\boldsymbol{x})+R(\boldsymbol{x}),
    \label{eq:bg_variational_problem}
\end{equation}
where $f$ measures consistency with the observations, $R$ favors plausible solutions, and $\eta>0$ controls the relative weight of the data-fidelity term.

\paragraph{Plug-and-Play Denoisers}
A common way to solve~\eqref{eq:bg_variational_problem} is to use first-order proximal splitting algorithms~\citep{beck2009fast,parikh2014proximal}. It handles the data-fidelity term by a gradient step and the regularizer by its proximal map. We define
\begin{equation}
    \operatorname{prox}_{R}(\boldsymbol{z})
    :=
    \arg\min_{\boldsymbol{x}\in\mathbb{R}^n}
    \left\{
        \frac{1}{2}\|\boldsymbol{x}-\boldsymbol{z}\|^2
        +
        R(\boldsymbol{x})
    \right\},
    \label{eq:bg_prox_def}
\end{equation}
where the proximal parameter is absorbed into \(R\). With the normalization in~\eqref{eq:bg_variational_problem}, proximal gradient descent takes the form
\begin{equation}
    \boldsymbol{z}_{k+1}
    =
    \boldsymbol{x}_k-\eta\nabla f(\boldsymbol{x}_k),
    \qquad
    \boldsymbol{x}_{k+1}
    =
    \operatorname{prox}_{R}(\boldsymbol{z}_{k+1}).
    \label{eq:bg_pgd_exact}
\end{equation}

\noindent
This splitting viewpoint naturally motivates Plug-and-Play (PnP) methods~\citep{venkatakrishnan2013plug,chan2017plug,kamilov2023plug}, which replace the proximal operator of \(R\) with an image
denoiser \(\mathsf{D}\):
\begin{equation}
    \boldsymbol{z}_{k+1}
    =
    \boldsymbol{x}_k-\eta\nabla f(\boldsymbol{x}_k),
    \qquad
    \boldsymbol{x}_{k+1}
    =
    \mathsf{D}(\boldsymbol{z}_{k+1}).
    \label{eq:bg_pnp_pgd}
\end{equation}

\noindent
The denoiser is commonly trained to remove additive white Gaussian noise and is then used as an implicit image prior~\citep{zhang2021plug,kamilov2023plug}. This design allows PnP-PGD to use learned priors without specifying $R$ explicitly. However, a generic denoiser need not be the proximal map of any regularizer, so the resulting iteration may not minimize an explicit objective. In this work, we focus on denoisers that admit a proximal representation~\citep{hertrich2021convolutional,hurault2022gradient,hurault2022proximal,fang2024whats}. This setting allows PnP-PGD to be analyzed as proximal gradient descent on an explicit objective.

\subsection{Structured Proximal Denoisers}
A denoiser $\mathsf D$ has a proximal representation if there exists a regularizer $R$ such that
$\mathsf D=\operatorname{prox}_{R}$. This structure allows PnP-PGD to be related to an explicit optimization objective. We represent the target domain by a regularizer $R_\star$ and define its reference proximal map as
\begin{equation}
    \mathsf{D}_\star
    =
    \operatorname{prox}_{R_\star}.
    \label{eq:bg_target_reference}
\end{equation}
The corresponding target objective is
\begin{equation}
    F_\star(\boldsymbol{x})
    =
    \eta f(\boldsymbol{x})+R_\star(\boldsymbol{x}).
    \label{eq:bg_target_objective}
\end{equation}
Using $\mathsf D_\star$ in~\eqref{eq:bg_pnp_pgd} recovers proximal gradient descent on $F_\star$ under the normalization introduced above.

\noindent
In practice, PnP-PGD uses a deployed denoiser $\widehat{\mathsf D}$. This denoiser may be trained on a source domain or adapted using limited target-domain data. Its outputs may therefore differ from those of $\mathsf D_\star$ at the query points generated by PnP-PGD. We study how this discrepancy affects stationarity of $F_\star$. This notation separates the target model from the deployed operator. The map $\mathsf D_\star$ defines the target regularizer and objective, while $\widehat{\mathsf D}$ is the operator used during reconstruction. The next two paragraphs describe two denoiser families that realize proximal maps in different ways: learned proximal networks and gradient-step denoisers~\citep{hurault2022gradient,hurault2022proximal,fang2024whats}.

\paragraph{Learned Proximal Networks.}
Learned Proximal Networks (LPNs) are designed so that the learned denoising map is an exact proximal operator of an associated learned regularizer~\citep{fang2024whats}. More specifically, they parameterize the denoiser through the gradient of a strongly convex potential, which yields a globally defined proximal map and an induced, generally nonconvex, regularization function. In our setting, this architecture-level proximality provides a direct realization of the target reference
\(\mathsf{D}_\star=\operatorname{prox}_{R_\star}\). Consequently, the target objective \(F_\star=\eta f+R_\star\) is not merely implicit in the iteration but is explicitly tied to the learned proximal map.

\paragraph{Gradient-Step Proximal Denoisers.}
A second family is given by Gradient-Step (GS) denoisers, which take the form
\begin{equation}
    \mathsf{D}_\theta
    =
    \operatorname{Id}-\nabla g_\theta,
    \label{eq:bg_gs_denoiser}
\end{equation}
where \(g_\theta\) is a scalar-valued neural potential~\citep{cohen2021has,hurault2022gradient}. This parameterization imposes a conservative structure on the denoiser, since the residual \(\operatorname{Id}-\mathsf{D}_\theta\) is represented as the gradient of a potential. Under suitable smoothness and contractivity conditions, such maps can be interpreted as proximal operators of induced regularizers~\citep{hurault2022proximal}. GS denoisers therefore provide another route to a well-defined proximal reference, but with a different guarantee mechanism from LPNs: LPNs build proximality directly into the architecture, whereas GS denoisers obtain proximal structure through analytic conditions on the learned potential.

\noindent
Together, these two examples define the proximal-structured setting considered in this paper. In both cases, one can distinguish between the ideal target proximal map \(\mathsf{D}_\star\) and the deployed inference denoiser
\(\widehat{\mathsf{D}}\). This distinction is the basis for the subsequent analysis of how denoiser mismatch, or prior shift, propagates into stationarity errors for the target objective \(F_\star\).

\section{Theory}
\label{sec:theory}
Our analysis consists of three main steps. First, we show how the mismatch between the deployed and target proximal denoisers affects the target proximal update. Second, we prove that the target objective still decreases, up to an error caused by this mismatch. Third, we show that the target proximal step approximately satisfies its optimality condition. Putting these results together, we obtain a bound on the average stationarity gap, where the remaining error depends on the mismatch gap. Finally, we show that the assumptions used in our analysis hold for Learned Proximal Networks and Gradient-Step denoisers.

\subsection{Target Objective and Basic Assumptions}
\label{subsec:general_target_prox}

All stationarity statements are made with respect to the target objective~\eqref{eq:bg_target_objective} defined by the target proximal reference
$\mathsf{D}_\star=\operatorname{prox}_{R_\star}.$
The map $\mathsf{D}_\star$ specifies the target prior and the objective whose stationary points are of interest, whereas the deployed denoiser $\widehat{\mathsf{D}}$ determines the iterates generated by PnP-PGD.

\noindent
We measure target stationarity through $\|\nabla F_\star(x)\|$. Accordingly, $R_\star$ is assumed to be differentiable at the deployed iterates. The additional regularity required for the target proximal-subproblem is stated below and verified for LPN and GS target references in Section~\ref{subsec:proximal_realizations}.

\begin{myassump}[Smooth data fidelity]
\label{assump:f_smooth}
The data-fidelity term \(f:\mathbb{R}^n\to\mathbb{R}\) is differentiable and has \(L_f\)-Lipschitz continuous gradient.
\end{myassump}

\begin{myassump}[Lower bounded target objective]
\label{assump:lower_bound}
The target objective \(F_\star\) is bounded below:
\[
    F_\star^{\inf}
    :=
    \inf_{\boldsymbol{x}}F_\star(\boldsymbol{x})
    >
    -\infty .
\]
\end{myassump}

\subsection{Proximal Mismatch as Inexact Proximal Step}
\label{subsec:proximal_mismatch_inexact_step}

Let $\{\boldsymbol{x}_k\}_{k\geq0}$ be the PnP-PGD sequence generated by the deployed denoiser $\widehat{\mathsf D}$. For $k\geq1$, define the query point
\begin{equation}
    \boldsymbol{z}_k
    =
    \boldsymbol{x}_{k-1}
    -
    \eta\nabla f(\boldsymbol{x}_{k-1}).
    \label{eq:zk_def}
\end{equation}
The deployed update is
\begin{equation}
    \boldsymbol{x}_k
    =
    \widehat{\mathsf D}(\boldsymbol{z}_k).
    \label{eq:deployed_update}
\end{equation}
At the same query point, the target proximal response is
$\mathsf D_\star(\boldsymbol{z}_k)$. We define the proximal error as
\begin{equation}
    \boldsymbol{\xi}_k
    :=
    \widehat{\mathsf D}(\boldsymbol{z}_k)
    -
    \mathsf D_\star(\boldsymbol{z}_k),
    \label{eq:ek_def}
\end{equation}
and the proximal mismatch as
\begin{equation}
    d_k
    :=
    \|\boldsymbol{\xi}_k\|_2
    =
    \big\|
        \widehat{\mathsf D}(\boldsymbol{z}_k)
        -
        \mathsf D_\star(\boldsymbol{z}_k)
    \big\|_2.
    \label{eq:dk_def}
\end{equation}
When a deterministic upper bound is available, we write
$d_k\leq\Delta_k$.

\noindent
The deployed update can now be written as
\begin{equation}
    \boldsymbol{x}_k
    =
    \mathsf D_\star(\boldsymbol{z}_k)
    +
    \boldsymbol{\xi}_k.
    \label{eq:inexact_target_prox_step}
\end{equation}
Thus, using $\widehat{\mathsf D}$ in place of $\mathsf D_\star$ gives an inexact proximal step for the target objective.

\noindent
To connect $d_k$ to standard measures of proximal inexactness, define the target proximal subproblem
\begin{equation}
    H_k(\boldsymbol{u})
    :=
    \frac{1}{2}
    \|\boldsymbol{u}-\boldsymbol{z}_k\|_2^2
    +
    R_\star(\boldsymbol{u}).
    \label{eq:Hk_def}
\end{equation}
Since
\[
    \mathsf D_\star(\boldsymbol{z}_k)
    \in
    \arg\min_{\boldsymbol{u}} H_k(\boldsymbol{u}),
\]
we define the proximal subproblem error as
\begin{equation}
    \epsilon_k
    :=
    H_k\big(
        \widehat{\mathsf D}(\boldsymbol{z}_k)
    \big)
    -
    H_k\big(
        \mathsf D_\star(\boldsymbol{z}_k)
    \big)
    \geq 0.
    \label{eq:eps_induced}
\end{equation}

\begin{mylem}[Proximal mismatch controls proximal subproblem error]
\label{lemma:bridge_discrepancy}
Fix $k\geq1$ and let $H_k$ be defined by~\eqref{eq:Hk_def}. Assume that $H_k$ is differentiable and $L_H$-smooth on the line segment between
$\mathsf D_\star(\boldsymbol{z}_k)$ and
$\widehat{\mathsf D}(\boldsymbol{z}_k)$. Then
\begin{equation}
    \epsilon_k
    \leq
    \frac{L_H}{2}
    \big\|
        \widehat{\mathsf D}(\boldsymbol{z}_k)
        -
        \mathsf D_\star(\boldsymbol{z}_k)
    \big\|_2^2
    =
    \frac{L_H}{2}d_k^2.
    \label{eq:bridge_smooth}
\end{equation}
Consequently, if $d_k\leq\Delta_k$, then
\begin{equation}
    \epsilon_k
    \leq
    \frac{L_H}{2}\Delta_k^2.
\end{equation}
\end{mylem}

\noindent\emph{Proof.}
See Appendix~\ref{app:proof-bridge}.

\subsection{Convergence Analysis}
\label{subsec:general_convergence}

The next assumption summarizes the regularity needed to convert proximal inexactness into a stationarity bound.

\begin{myassump}[Proximal subproblem regularity]
\label{assump:prox_subproblem_regularity}
For every visited query point \(\boldsymbol{z}_k\), the target proximal
subproblem \(H_k\) has the unique minimizer
\(\mathsf{D}_\star(\boldsymbol{z}_k)\). Moreover, \(H_k\) satisfies
Lemma~\ref{lemma:bridge_discrepancy} with a constant \(L_H\) independent of
\(k\), and at the deployed iterate
\(\boldsymbol{x}_k=\widehat{\mathsf{D}}(\boldsymbol{z}_k)\),
\begin{equation}
    \|\nabla H_k(\boldsymbol{x}_k)\|^2
    \leq
    2L_H
    \left(
        H_k(\boldsymbol{x}_k)
        -
        H_k(\mathsf{D}_\star(\boldsymbol{z}_k))
    \right).
    \label{eq:gradient_suboptimality}
\end{equation}
\end{myassump}

\noindent
In the LPN and GS realizations below, this condition follows because the proximal subproblem can be written, up to constants, as a smooth convex potential minus a linear term. In particular,
\[
    \nabla H_k(\boldsymbol{u})
    =
    \mathsf{D}_\star^{-1}(\boldsymbol{u})
    -
    \boldsymbol{z}_k,
\]
and the Lipschitz regularity of \(\mathsf{D}_\star^{-1}\) provides a uniform
smoothness constant.

\begin{mylem}[Descent with inexact target-proximal steps]
\label{lemma:abstract_descent}
Under Assumption~\ref{assump:f_smooth}, the deployed PnP-PGD iterates satisfy
\begin{equation}
    F_\star(\boldsymbol{x}_{k+1})
    \leq
    F_\star(\boldsymbol{x}_k)
    -
    \frac{1-\eta L_f}{2}
    \|\boldsymbol{x}_{k+1}-\boldsymbol{x}_k\|^2
    +
    \epsilon_{k+1}.
    \label{eq:abstract_descent}
\end{equation}
\end{mylem}

\noindent\emph{Proof.}
See Appendix~\ref{app:proof-descent}.

\begin{mylem}[Stationarity residual control]
\label{lemma:abstract_residual}
Under Assumptions~\ref{assump:prox_subproblem_regularity}, for each \(k\geq0\) there exists a residual vector \(\boldsymbol{r}_{k+1}\) such that
\begin{equation}
    \nabla R_\star(\boldsymbol{x}_{k+1})
    =
    \boldsymbol{z}_{k+1}
    -
    \boldsymbol{x}_{k+1}
    -
    \boldsymbol{r}_{k+1},
    \qquad
    \|\boldsymbol{r}_{k+1}\|^2
    \leq
    2L_H\epsilon_{k+1}.
    \label{eq:abstract_residual}
\end{equation}
\end{mylem}

\noindent\emph{Proof.}
See Appendix~\ref{app:proof-residual}.

\begin{mythm}[Convergence under proximal mismatch]
\label{thm:main_stationarity}
Suppose that Assumptions~\ref{assump:f_smooth}--\ref{assump:prox_subproblem_regularity}
hold and that \(\eta L_f<1\). Define
\begin{equation}
    C_0
    :=
    \frac{16}{1-\eta L_f}
    \big(F_\star(\boldsymbol{x}_0)-F_\star^{\inf}\big),
    \qquad
    C_1
    :=
    \left(
        \frac{16}{1-\eta L_f}
        +
        4L_H
    \right)
    \frac{L_H}{2}.
    \label{eq:constants_main_bound}
\end{equation}
Then, for all \(t\geq1\),
\begin{equation}
    \frac{1}{t}
    \sum_{k=1}^{t}
    \|\nabla F_\star(\boldsymbol{x}_k)\|^2
    \leq
    \frac{C_0}{t}
    +
    C_1
    \frac{1}{t}
    \sum_{k=1}^{t}d_k^2 .
    \label{eq:abstract_rate}
\end{equation}
Consequently,
\[
    \min_{1\leq k\leq t}
    \|\nabla F_\star(\boldsymbol{x}_k)\|^2
    \leq
    \frac{C_0}{t}
    +
    C_1
    \frac{1}{t}
    \sum_{k=1}^{t}d_k^2 .
\]
If \(d_k\leq\Delta_k\), the same bounds hold with \(d_k\) replaced by
\(\Delta_k\). If \(\sum_{k=1}^{\infty}d_k^2<\infty\), then
\[
    \|\nabla F_\star(\boldsymbol{x}_k)\|\to0 .
\]
In particular, this holds whenever \(d_k^2=O(k^{-1-\nu})\) for some
\(\nu>0\).
\end{mythm}

\noindent\emph{Proof.}
See Appendix~\ref{app:proof-main}.

\begin{mycor}[Matched target proximal denoiser]
\label{cor:exact_target}
If \(\widehat{\mathsf{D}}=\mathsf{D}_\star\), then \(d_k=0\) for all \(k\), and
Theorem~\ref{thm:main_stationarity} gives
\begin{equation}
    \frac{1}{t}
    \sum_{k=1}^{t}
    \|\nabla F_\star(\boldsymbol{x}_k)\|^2
    \leq
    \frac{16}{(1-\eta L_f)t}
    \big(
        F_\star(\boldsymbol{x}_0)-F_\star^{\inf}
    \big).
    \label{eq:exact_rate}
\end{equation}
\end{mycor}

\noindent\emph{Proof.}
See Appendix~\ref{app:proof-exact-target}.

\subsection{Structured Proximal-Denoiser Realizations}
\label{subsec:proximal_realizations}

The preceding result applies whenever the target prior admits a proximal reference satisfying the regularity conditions above. We instantiate this abstract setting with two structured proximal denoiser families: Learned Proximal Networks, which provide architecture-level proximal maps, and Gradient-Step proximal denoisers, which provide proximal maps under contractivity conditions. Detailed proximal
characterizations and derivations are provided in
Appendix~\ref{app:structured_proximal_realizations}.

\subsubsection{Learned Proximal Networks}

Following~\citet{fang2024whats}, an LPN is parameterized as
\begin{equation}
    \mathsf{D}_\theta(\boldsymbol{x})
    =
    \nabla\Psi_\theta(\boldsymbol{x}),
    \qquad
    \Psi_\theta(\boldsymbol{x})
    =
    \psi_\theta(\boldsymbol{x})
    +
    \frac{\alpha}{2}\|\boldsymbol{x}\|^2,
    \qquad
    \alpha\in(0,1),
    \label{eq:LPN_def}
\end{equation}
where \(\psi_\theta\) is an input-convex neural network with
\(\mathcal{C}^2\) activations. Such a network is an exact proximal map
\(\mathsf{D}_\theta=\operatorname{prox}_{R_\theta}\), and its inverse is
\(1/\alpha\)-Lipschitz; see
Proposition~\ref{prop:lpn} in the appendix.

\noindent
For the target denoiser
\(\mathsf{D}_\star=\mathsf{D}_{\theta^\star}\), the associated proximal
subproblem satisfies, up to additive constants,
\[
    H_k(\boldsymbol{u})
    =
    \Psi_{\theta^\star}^\star(\boldsymbol{u})
    -
    \langle\boldsymbol{z}_k,\boldsymbol{u}\rangle.
\]
Consequently, \(H_k\) is convex and \(1/\alpha\)-smooth. Therefore,
Assumption~\ref{assump:prox_subproblem_regularity} holds with
\[
    L_H=\frac{1}{\alpha},
\]
and Lemma~\ref{lemma:bridge_discrepancy} yields
\begin{equation}
    \epsilon_k
    \leq
    \frac{1}{2\alpha}
    \big\|
        \widehat{\mathsf{D}}_\theta(\boldsymbol{z}_k)
        -
        \mathsf{D}_{\theta^\star}(\boldsymbol{z}_k)
    \big\|^2.
    \label{eq:lpn_eps}
\end{equation}
Hence Theorem~\ref{thm:main_stationarity} applies with
\[
    d_k
    =
    \big\|
        \widehat{\mathsf{D}}_\theta(\boldsymbol{z}_k)
        -
        \mathsf{D}_{\theta^\star}(\boldsymbol{z}_k)
    \big\|.
\]

\subsubsection{Gradient-Step Proximal Denoisers}

GS denoisers take the form
\begin{equation}
    \mathsf{D}_\sigma
    =
    \nabla h_\sigma
    =
    \operatorname{Id}-\nabla g_\sigma,
    \qquad
    h_\sigma(\boldsymbol{x})
    =
    \frac{1}{2}\|\boldsymbol{x}\|^2-g_\sigma(\boldsymbol{x}).
    \label{eq:gs_def}
\end{equation}
Suppose that \(\nabla g_{\sigma,\star}\) is \(L\)-Lipschitz with \(L<1\).
Then the target GS denoiser
\(\mathsf{D}_{\sigma,\star}\) is a proximal map on its image, and its
inverse is \(1/(1-L)\)-Lipschitz; see
Proposition~\ref{prop:gs_prox} in the appendix.

\noindent
Assume additionally that
\[
    \widehat{\mathsf{D}}_\sigma(\boldsymbol{z}_k)
    \in
    \operatorname{Im}(\mathsf{D}_{\sigma,\star})
\]
at every visited query point. On this set, the associated proximal
subproblem satisfies, up to additive constants,
\[
    H_k(\boldsymbol{u})
    =
    h_{\sigma,\star}^\star(\boldsymbol{u})
    -
    \langle\boldsymbol{z}_k,\boldsymbol{u}\rangle.
\]
Thus \(H_k\) is convex and \(1/(1-L)\)-smooth, so that
\[
    L_H=\frac{1}{1-L}.
\]
Accordingly,
\begin{equation}
    \epsilon_k
    \leq
    \frac{1}{2(1-L)}
    \big\|
        \widehat{\mathsf{D}}_\sigma(\boldsymbol{z}_k)
        -
        \mathsf{D}_{\sigma,\star}(\boldsymbol{z}_k)
    \big\|^2,
    \label{eq:gs_eps}
\end{equation}
and Theorem~\ref{thm:main_stationarity} applies with
\[
    d_k
    =
    \big\|
        \widehat{\mathsf{D}}_\sigma(\boldsymbol{z}_k)
        -
        \mathsf{D}_{\sigma,\star}(\boldsymbol{z}_k)
    \big\|.
\]
LPNs and GS denoisers therefore instantiate the same proximal-mismatch
theorem through different mechanisms: LPNs yield globally defined proximal
maps by construction, whereas GS denoisers yield proximal maps on the image
of the target denoiser under the condition \(L<1\).

\section{Proximal Matching Adaptation}
\label{sec:adaptation_strategies}

Theorem~\ref{thm:main_stationarity} shows that target stationarity depends on the average deployed-to-target proximal discrepancy along the PnP-PGD trajectory. This result motivates adapting the deployed denoiser toward the target proximal map. In practice, direct evaluations of the target proximal map are unavailable during few-shot adaptation. We instead use a clean target-domain set
\[
    \mathcal S_{\mathrm{tar}}^n
    =
    \{\boldsymbol{x}_i\}_{i=1}^{n_{\mathrm{adapt}}}
\]
and generate Gaussian-corrupted inputs
\begin{equation}
    \boldsymbol{y}_{i,j}
    =
    \boldsymbol{x}_i
    +
    \sigma_{\mathrm{adapt}}\boldsymbol{\delta}_{i,j},
    \qquad
    \boldsymbol{\delta}_{i,j}\sim\mathcal N(0,I),
    \qquad
    j=1,\ldots,M.
    \label{eq:adaptation_inputs}
\end{equation}
We first compare the population targets of squared-error regression and proximal matching. We then apply proximal matching to LPN and GS denoisers.

\subsection{Population Targets of MSE and Proximal Matching}
\label{subsec:population_targets}

Let
\[
    X\sim p_{\mathrm{tar}},
    \qquad
    Y=X+\sigma_{\mathrm{adapt}}\delta,
    \qquad
    \delta\sim\mathcal N(0,I).
\]
Then
\[
    p_{\mathrm{tar}}(x\mid y)
    \propto
    \exp
    \left(
        -\frac{\|y-x\|_2^2}{2\sigma_{\mathrm{adapt}}^2}
    \right)
    p_{\mathrm{tar}}(x).
\]
For \(R_{\mathrm{tar}}(x)=-\log p_{\mathrm{tar}}(x)\), the MAP denoiser is
\begin{equation}
    \mathsf{D}_{\mathrm{MAP}}(y)
    \in
    \arg\min_x
    \left\{
        \frac{1}{2}\|x-y\|_2^2
        +
        \sigma_{\mathrm{adapt}}^2R_{\mathrm{tar}}(x)
    \right\}
    =
    \operatorname{prox}_{\sigma_{\mathrm{adapt}}^2R_{\mathrm{tar}}}(y).
    \label{eq:map_as_prox}
\end{equation}

\noindent
For \(u,v\in\mathbb R^d\), define
\begin{equation}
    \rho_\gamma(u,v)
    =
    1-\exp
    \left(
        -\frac{\|u-v\|_2^2}{2\gamma^2}
    \right),
    \qquad
    \gamma>0.
    \label{eq:pm_pairwise_loss}
\end{equation}

\begin{myprop}[Population targets]
\label{prop:pm_vs_mse}
Let \(\mathsf{D}:\mathbb R^d\to\mathbb R^d\) range over all measurable
denoisers, and assume that \(\mathbb E\|X\|_2^2<\infty\). The population MSE
risk
\begin{equation}
    \mathcal R_{\mathrm{MSE}}(\mathsf{D})
    =
    \mathbb E
    \left[
        \|\mathsf{D}(Y)-X\|_2^2
    \right]
    \label{eq:population_mse_risk}
\end{equation}
is minimized by
\begin{equation}
    \mathsf{D}_{\mathrm{MSE}}^\star(y)
    =
    \mathbb E[X\mid Y=y].
    \label{eq:posterior_mean_target}
\end{equation}
The population PM risk
\begin{equation}
    \mathcal R_{\mathrm{PM},\gamma}(\mathsf{D})
    =
    \mathbb E
    \left[
        \rho_\gamma(\mathsf{D}(Y),X)
    \right]
    \label{eq:population_pm_risk}
\end{equation}
has conditional minimizers
\begin{equation}
    \mathsf{D}_{\mathrm{PM},\gamma}^\star(y)
    \in
    \arg\max_z
    \mathbb E
    \left[
        \exp
        \left(
            -\frac{\|z-X\|_2^2}{2\gamma^2}
        \right)
        \,\middle|\,
        Y=y
    \right].
    \label{eq:kernel_smoothed_posterior}
\end{equation}
Thus, PM targets the mode of a Gaussian-smoothed posterior. Under the
consistency conditions of~\citet{fang2024whats}, if the posterior mode is
unique, then
\begin{equation}
    \mathsf{D}_{\mathrm{PM},\gamma}^\star(y)
    \longrightarrow
    \mathsf{D}_{\mathrm{MAP}}(y)
    \qquad
    \text{as }
    \gamma\downarrow0.
    \label{eq:pm_mode_target}
\end{equation}
\end{myprop}

\begin{proof}
See Appendix~\ref{proof:prop_pm_vs_mse}.
\end{proof}
MSE targets the posterior mean, whereas PM approaches the MAP/proximal
denoiser. This motivates using PM to reduce proximal mismatch.

\subsection{LPN Adaptation}
\label{subsec:lpn_adaptation}

Let \(\mathsf{D}_{\theta_s}\) be a source-trained LPN. We fine-tune the same
proximal parameterization,
\[
    \mathsf{D}_\theta=\nabla\Psi_\theta,
\]
on the target adaptation set.

\paragraph{Clean-MSE adaptation.}
The MSE baseline minimizes
\begin{equation}
    \mathcal L_{\mathrm{MSE}}(\theta)
    =
    \frac{1}{n_{\mathrm{adapt}}M}
    \sum_{i=1}^{n_{\mathrm{adapt}}}
    \sum_{j=1}^{M}
    \left\|
        \mathsf{D}_\theta(\boldsymbol{y}_{i,j})
        -
        \boldsymbol{x}_i
    \right\|_2^2.
    \label{eq:mse_adaptation_loss}
\end{equation}
Starting from \(\theta_s\), we obtain
\begin{equation}
    \theta_{\mathrm{MSE}}
    \in
    \arg\min_{\theta\in\Theta_{\mathrm{LPN}}}
    \mathcal L_{\mathrm{MSE}}(\theta).
    \label{eq:mse_adapted_parameter}
\end{equation}

\paragraph{Proximal-matching adaptation.}
PM minimizes
\begin{equation}
    \mathcal L_{\mathrm{PM}}(\theta;\gamma)
    =
    \frac{1}{n_{\mathrm{adapt}}M}
    \sum_{i=1}^{n_{\mathrm{adapt}}}
    \sum_{j=1}^{M}
    \rho_\gamma
    \left(
        \mathsf{D}_\theta(\boldsymbol{y}_{i,j}),
        \boldsymbol{x}_i
    \right).
    \label{eq:pm_adaptation_loss}
\end{equation}
Starting from \(\theta_s\), we obtain
\begin{equation}
    \theta_{\mathrm{PM}}
    \in
    \arg\min_{\theta\in\Theta_{\mathrm{LPN}}}
    \mathcal L_{\mathrm{PM}}(\theta;\gamma).
    \label{eq:pm_adapted_parameter}
\end{equation}
The bandwidth \(\gamma\) may be annealed during fine-tuning. Both objectives
preserve the LPN parameterization but have different population targets.

\subsection{Gradient-Step Proximal Adaptation}
\label{subsec:gs_adaptation}

A GS denoiser has the form
\begin{equation}
    \mathsf{D}_\theta^{\mathrm{GS}}(\boldsymbol{u})
    =
    \boldsymbol{u}
    -
    \nabla g_\theta(\boldsymbol{u}),
    \label{eq:gs_denoiser_adapt}
\end{equation}
where \(g_\theta:\mathbb R^d\to\mathbb R\) is differentiable. Its proximal
interpretation requires \(\nabla g_\theta\) to be contractive.

\paragraph{Clean-MSE adaptation.}
The GS MSE baseline is
\begin{equation}
    \mathcal L_{\mathrm{GS\text{-}MSE}}(\theta)
    =
    \frac{1}{n_{\mathrm{adapt}}M}
    \sum_{i=1}^{n_{\mathrm{adapt}}}
    \sum_{j=1}^{M}
    \left\|
        \mathsf{D}_\theta^{\mathrm{GS}}(\boldsymbol{y}_{i,j})
        -
        \boldsymbol{x}_i
    \right\|_2^2.
    \label{eq:gs_mse}
\end{equation}

\paragraph{Gradient-field proximal matching.}
Because
\[
    \mathsf{D}_\theta^{\mathrm{GS}}(\boldsymbol{y}_{i,j})
    -
    \boldsymbol{x}_i
    =
    -
    \left[
        \nabla g_\theta(\boldsymbol{y}_{i,j})
        -
        (\boldsymbol{y}_{i,j}-\boldsymbol{x}_i)
    \right],
\]
output matching is equivalent to matching the gradient field to the corruption
residual. We define
\begin{equation}
    \mathcal L_{\mathrm{GS\text{-}PM}}(\theta;\gamma)
    =
    \frac{1}{n_{\mathrm{adapt}}M}
    \sum_{i=1}^{n_{\mathrm{adapt}}}
    \sum_{j=1}^{M}
    \rho_\gamma
    \left(
        \nabla g_\theta(\boldsymbol{y}_{i,j}),
        \boldsymbol{y}_{i,j}-\boldsymbol{x}_i
    \right).
    \label{eq:gs_pm}
\end{equation}

\paragraph{Contractivity penalty.}
Let \(\widehat L(\nabla g_\theta)\) estimate the Lipschitz constant on the
adaptation inputs. For \(L_{\max}<1\), define
\begin{equation}
    \mathcal R_{\mathrm{con}}(\theta)
    =
    \left[
        \widehat L(\nabla g_\theta)-L_{\max}
    \right]_+^2.
    \label{eq:gs_contractivity_regularizer}
\end{equation}
The final objective is
\begin{equation}
    \mathcal L_{\mathrm{GS\text{-}AdaPM}}(\theta)
    =
    \mathcal L_{\mathrm{GS\text{-}PM}}(\theta;\gamma)
    +
    \lambda_{\mathrm{con}}
    \mathcal R_{\mathrm{con}}(\theta).
    \label{eq:gs_adapm}
\end{equation}
The PM term adapts the gradient field. The penalty promotes the contractivity
condition required by the GS proximal interpretation.

\section{Numerical Experiments}

Our experiments address three questions. First, how much does deploying a source-trained denoiser degrade PnP-PGD reconstruction on the target domain? Second, in the few-shot regime, does proximal-matching adaptation outperform MSE-based adaptation? Third, are reconstruction gains accompanied by lower proximal mismatch along the PnP-PGD trajectory?

\noindent
We study these questions under a substantial BreCaHAD~\citep{liu2015deep}-to-CelebA~\citep{aksac2019brecahad} domain shift for Gaussian deblurring and single-image super-resolution. We consider two structured denoiser families: LPNs and GS denoisers. For each family, we compare the deployed denoiser with a target-trained proximal reference from the same family and evaluate both reconstruction performance and on-trajectory proximal mismatch.

\subsection{Experimental Setup}

\paragraph{Data and domain shift.}
We train the source denoisers on BreCaHAD and perform few-shot adaptation and reconstruction on CelebA~\citep{aksac2019brecahad,liu2015deep}. BreCaHAD contains breast histopathology images with cellular morphology and stain-dependent appearance, whereas CelebA contains aligned face images. This source-to-target shift changes both semantic content and low-level image statistics.

\noindent
For each inverse problem, we evaluate all methods on the same 100 images from the CelebA test set. The target-domain adaptation set is disjoint from the test set and contains
\[
    n_{\mathrm{adapt}}
    \in
    \{1,5,25,50\}
\]
images. We report the mean and standard deviation of the reconstruction metrics across the 100 test images. Figure~\ref{fig:stage2-source-target-domain-gap} shows representative images from the source and target domains.

\begin{figure}[t]
\centering
\includegraphics[width=0.75\linewidth]{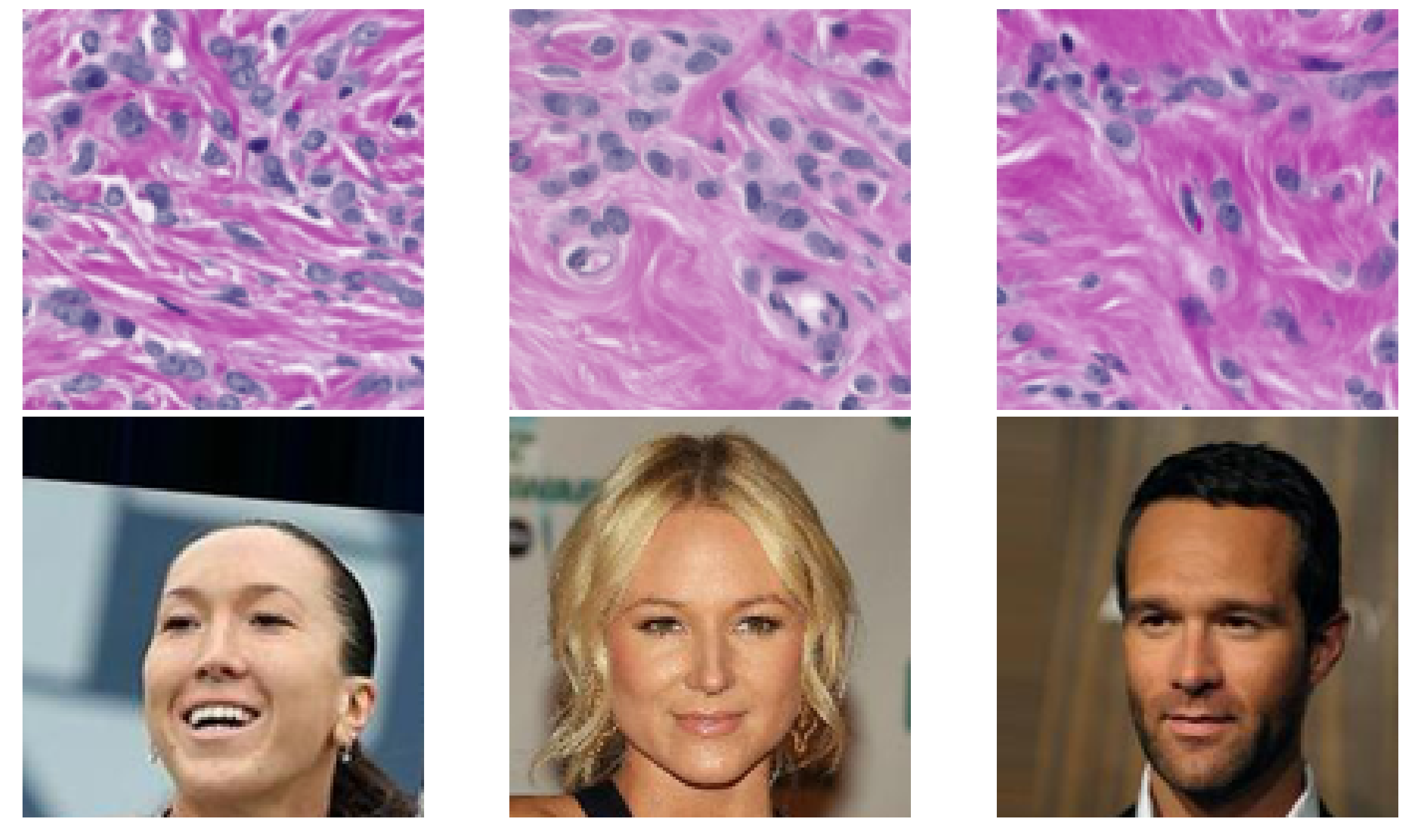}
\caption{Representative images from the source and target domains. BreCaHAD contains histopathology patches, whereas CelebA contains aligned facial images, producing a substantial source-target shift in both semantic content and visual statistics.}
\label{fig:stage2-source-target-domain-gap}
\end{figure}

\paragraph{Gaussian deblurring.}
We generate blurred observations from RGB images
$\boldsymbol{x}\in[0,1]^{3\times128\times128}$ according to
\begin{equation}
    \boldsymbol{y}
    =
    A_{\mathrm{db}}\boldsymbol{x}
    +
    \boldsymbol{e},
    \qquad
    A_{\mathrm{db}}
    =
    H_g,
    \label{eq:deblurring-forward-model}
\end{equation}
where $H_g$ applies a $5\times5$ Gaussian kernel with standard deviation
$1.0$ independently to each color channel, and
$\boldsymbol{e}\sim\mathcal N(\boldsymbol{0},0.02^2I)$.
We use the data-fidelity term
\[
    f_{\mathrm{db}}(\boldsymbol{z})
    =
    \frac{1}{2}
    \|A_{\mathrm{db}}\boldsymbol{z}-\boldsymbol{y}\|_2^2.
\]
For the implemented operator,
$L_{\mathrm{db}}=\|A_{\mathrm{db}}\|_2^2=1$.
We run PnP-PGD for $40$ iterations with
$\eta_{\mathrm{db}}=0.95$ and initialize
$\boldsymbol{x}_0=A_{\mathrm{db}}^\top\boldsymbol{y}$.
Thus, $\eta_{\mathrm{db}}L_{\mathrm{db}}=0.95<1$.

\paragraph{Single-image super-resolution.}
For $\times4$ super-resolution, we use
\begin{equation}
    \boldsymbol{y}
    =
    A_{\mathrm{sr}}\boldsymbol{x},
    \qquad
    A_{\mathrm{sr}}
    =
    S_4H_{\sigma_{\mathrm{sr}}},
    \label{eq:sr-forward-model}
\end{equation}
where $H_{\sigma_{\mathrm{sr}}}$ applies a normalized
$25\times25$ Gaussian anti-aliasing kernel with
$\sigma_{\mathrm{sr}}=1.6$ independently to each color channel, and
$S_4$ downsamples the result by a factor of four in each spatial dimension.
This operator maps an RGB $128\times128$ image to an RGB $32\times32$
measurement. We add no measurement noise. The data-fidelity term is
\[
    f_{\mathrm{sr}}(\boldsymbol{z})
    =
    \frac{1}{2}
    \|A_{\mathrm{sr}}\boldsymbol{z}-\boldsymbol{y}\|_2^2.
\]
We initialize PnP-PGD by bicubic upsampling followed by pixel-shift correction.
For the implemented operator,
$L_{\mathrm{sr}}=\|A_{\mathrm{sr}}\|_2^2\approx0.063$.
We run $40$ iterations with $\eta_{\mathrm{sr}}=15$, giving
$\eta_{\mathrm{sr}}L_{\mathrm{sr}}\approx0.945<1$.

\paragraph{Denoiser families and adaptation baselines.}
We evaluate two structured denoiser families. LPNs enforce proximal structure by construction, whereas GS denoisers admit a proximal interpretation under contractivity conditions on the learned gradient field.

\noindent
For each family, we consider four variants: a source denoiser trained on BreCaHAD, a target-reference denoiser trained on CelebA, an MSE-adapted source denoiser, and a proximal-matching-adapted source denoiser. The target-reference model serves as the proximal reference when evaluating on-trajectory mismatch.

\noindent
For GS, proximal-matching adaptation also includes the contractivity regularizer defined in Section~\ref{subsec:gs_adaptation}; we denote this model by GS-AdaPM. Within each family and adaptation budget, the MSE and proximal-matching variants use the same target images, corruption process, source initialization, and number of adaptation epochs. LPN-MSE and LPN-PM are adapted for 200 epochs, whereas GS-MSE and GS-AdaPM are adapted for 20 epochs. For each inverse problem, all denoisers are evaluated with the same PnP-PGD reconstruction settings.

\paragraph{Evaluation metrics.}
We report PSNR and SSIM for reconstruction quality. To measure proximal
mismatch, we also compute a family-matched trajectory gap. For a candidate
denoiser $\mathsf{D}$ and the target-reference denoiser $\mathsf{D}_{\theta^\star}$, we evaluate
both denoisers on the same PnP query $z_k$ at each iteration and compute
\[
\mathrm{gap}(z_k;\mathsf{D},\mathsf{D}_{\theta^\star})
=
\frac{\|\mathsf{D}(z_k)-\mathsf{D}_{\theta^\star}(z_k)\|_2}
{\max(\|\mathsf{D}_{\theta^\star}(z_k)\|_2,10^{-12})}.
\]
The reported gap is averaged over PnP iterations and test images. This metric
is always interpreted within the same denoiser family, i.e., LPN gaps compare
LPN denoisers to the LPN target reference and GS gaps compare GS denoisers to
the GS target reference. For GS models, we also report empirical local Jacobian diagnostics for $\nabla g_\theta$.

\subsection{Main Deblurring Results}

Table~\ref{tab:deblurring_lpn_gs_main} and
Figure~\ref{fig:deblurring-seed0-2x2} report the Gaussian deblurring results under the BreCaHAD-to-CelebA domain shift. The source-trained models show a large loss in reconstruction quality. LPN-source and GS-source achieve
$11.52$ and $19.75$ dB, respectively, compared with
$31.40$ and $31.75$ dB for the corresponding target-trained references. The source models also exhibit nonzero relative trajectory mismatch, showing that the reconstruction loss is accompanied by disagreement with the target proximal responses.

\begin{figure}[htp!]
    \centering
    \includegraphics[width=0.95\linewidth]{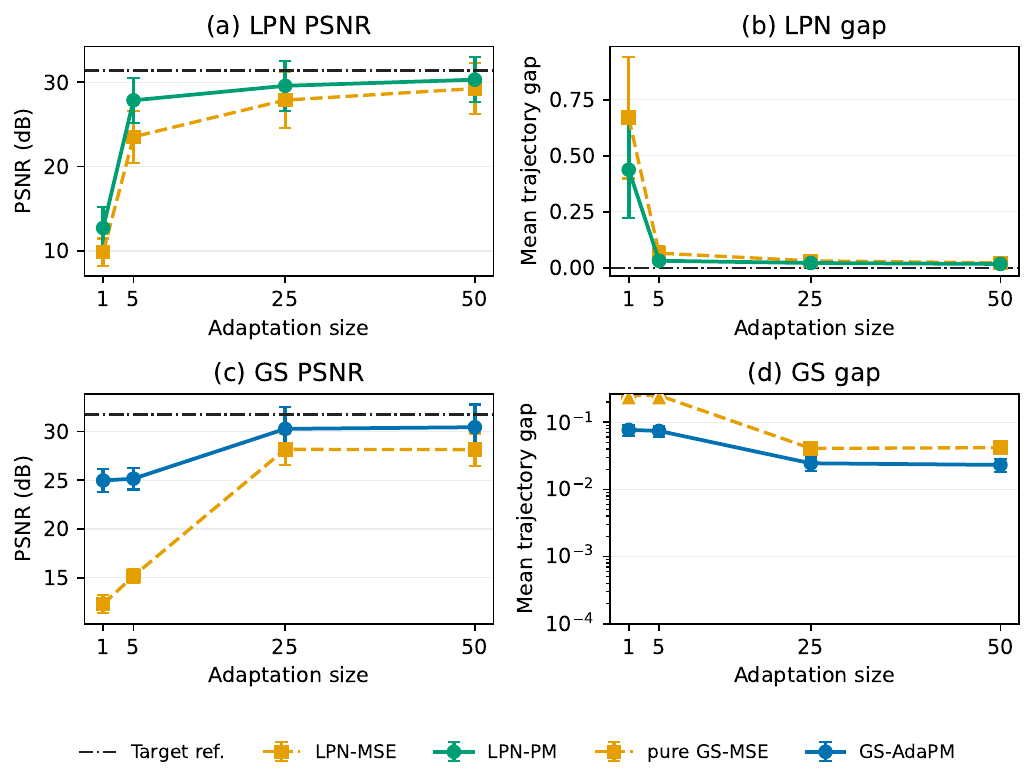}
    \caption{Gaussian deblurring performance under the BreCaHAD-to-CelebA domain shift.
    PSNR and relative trajectory mismatch are shown as functions of the number of target-domain adaptation images.
    Proximal-matching adaptation yields the largest PSNR gains in the low-data regime and reduces the mismatch with the corresponding target-trained reference.
    Higher PSNR and lower mismatch are better.}
    \label{fig:deblurring-seed0-2x2}
\end{figure}

\begin{table*}[t]
\centering
\caption{
Gaussian deblurring under severe BreCaHAD-to-CelebA proximal mismatch.
Values are mean $\pm$ standard deviation over 100 test images.
SSIM and trajectory gap (Gap) are reported in units of $10^{-2}$.
The gap is computed against the target-trained reference within each
denoiser family.
$\widehat L_{\max}$ denotes the empirical local Jacobian estimate
for GS denoisers.
}
\label{tab:deblurring_lpn_gs_main}

\sisetup{
    text-series-to-math       = true,
    uncertainty-mode          = separate,
    table-align-uncertainty   = true,
    allow-uncertainty-breaks  = false,
    table-alignment-mode      = format,
    table-number-alignment    = center,
    table-text-alignment      = center,
    table-model-setup         = \bfseries,
    group-digits              = none
}

\small
\setlength{\tabcolsep}{4.8pt}
\renewcommand{\arraystretch}{1.12}

\begin{tabular}{
    @{}
    l
    c
    S[table-format = 2.2(1.2)]
    S[table-format = 2.1(2.1)]
    S[table-format = 2.2(2.2)]
    S[table-format = 1.3]
    @{}
}
\toprule
\textbf{Method}
&
\textbf{$n_{\mathrm{adapt}}$}
&
{\textbf{PSNR $\uparrow$}}
&
{\textbf{SSIM $(10^{-2})$ $\uparrow$}}
&
{\textbf{Gap $(10^{-2})$ $\downarrow$}}
&
{\textbf{$\widehat L_{\max}$ $\downarrow$}}
\\
\midrule

\multicolumn{6}{@{}l}{\textit{LPN proximal denoiser}} \\
\cmidrule(lr){1-6}

LPN-source
& \multicolumn{1}{c}{--}
& 11.52(2.36)
& 37.3(18.9)
& 82.90(43.84)
& \multicolumn{1}{c}{--}
\\

\rowcolor{referencegray}
LPN-reference
& \multicolumn{1}{c}{--}
& 31.40(2.89)
& 89.7(4.2)
& \multicolumn{1}{c}{0}
& \multicolumn{1}{c}{--}
\\

\addlinespace[2pt]
LPN-MSE
&
& 9.86(1.61)
& 32.8(14.3)
& 67.01(27.15)
& \multicolumn{1}{c}{--}
\\

\rowcolor{methodgray}
\textbf{LPN-PM}
& \multirow{-2}{*}{$1$}
& \bfseries 12.72(2.43)
& \bfseries 50.5(14.6)
& \bfseries 43.77(21.49)
& \multicolumn{1}{c}{--}
\\

\addlinespace[2pt]
LPN-MSE
&
& 23.52(3.09)
& 77.5(6.9)
& 6.51(2.99)
& \multicolumn{1}{c}{--}
\\

\rowcolor{methodgray}
\textbf{LPN-PM}
& \multirow{-2}{*}{$5$}
& \bfseries 27.86(2.66)
& \bfseries 84.8(4.8)
& \bfseries 3.12(0.98)
& \multicolumn{1}{c}{--}
\\

\addlinespace[2pt]
LPN-MSE
&
& 27.89(3.36)
& 86.0(4.4)
& 3.08(1.56)
& \multicolumn{1}{c}{--}
\\

\rowcolor{methodgray}
\textbf{LPN-PM}
& \multirow{-2}{*}{$25$}
& \bfseries 29.58(2.94)
& \bfseries 87.8(4.1)
& \bfseries 2.11(1.03)
& \multicolumn{1}{c}{--}
\\

\addlinespace[2pt]
LPN-MSE
&
& 29.25(3.01)
& 87.4(4.1)
& 2.18(1.07)
& \multicolumn{1}{c}{--}
\\

\rowcolor{methodgray}
\textbf{LPN-PM}
& \multirow{-2}{*}{$50$}
& \bfseries 30.31(2.70)
& \bfseries 88.5(4.1)
& \bfseries 1.69(0.78)
& \multicolumn{1}{c}{--}
\\

\addlinespace[4pt]
\multicolumn{6}{@{}l}{\textit{GS proximal denoiser}} \\
\cmidrule(lr){1-6}

GS-source
& \multicolumn{1}{c}{--}
& 19.75(1.09)
& 49.9(8.6)
& 15.50(4.73)
& 1.039
\\

\rowcolor{referencegray}
GS-reference
& \multicolumn{1}{c}{--}
& 31.75(2.89)
& 91.4(3.4)
& \multicolumn{1}{c}{0}
& 0.927
\\

\addlinespace[2pt]
GS-MSE
&
& 12.29(0.94)
& 16.8(4.1)
& 48.68(11.28)
& 6.567
\\

\rowcolor{methodgray}
\textbf{GS-AdaPM}
& \multirow{-2}{*}{$1$}
& \bfseries 24.98(1.14)
& \bfseries 69.3(3.6)
& \bfseries 7.70(1.49)
& \bfseries 1.446
\\

\addlinespace[2pt]
GS-MSE
&
& 15.17(0.74)
& 25.0(4.5)
& 32.24(7.31)
& 5.077
\\

\rowcolor{methodgray}
\textbf{GS-AdaPM}
& \multirow{-2}{*}{$5$}
& \bfseries 25.17(1.11)
& \bfseries 70.0(3.5)
& \bfseries 7.41(1.31)
& \bfseries 1.005
\\

\addlinespace[2pt]
GS-MSE
&
& 28.18(1.61)
& 81.0(3.4)
& 4.08(0.85)
& 0.989
\\

\rowcolor{methodgray}
\textbf{GS-AdaPM}
& \multirow{-2}{*}{$25$}
& \bfseries 30.28(2.28)
& \bfseries 88.2(3.0)
& \bfseries 2.44(0.54)
& \bfseries 0.985
\\

\addlinespace[2pt]
GS-MSE
&
& 28.15(1.65)
& 81.1(3.4)
& 4.17(0.84)
& 0.992
\\

\rowcolor{methodgray}
\textbf{GS-AdaPM}
& \multirow{-2}{*}{$50$}
& \bfseries 30.45(2.33)
& \bfseries 88.7(3.0)
& \bfseries 2.31(0.52)
& \bfseries 0.986
\\

\bottomrule
\end{tabular}
\end{table*}

\noindent
For LPNs, proximal matching achieves higher PSNR and lower relative trajectory mismatch than MSE adaptation at every adaptation budget. The PSNR gains over LPN-MSE are
$2.86$, $4.34$, $1.69$, and $1.06$ dB for
$n_{\mathrm{adapt}}\in\{1,5,25,50\}$, respectively.
With one target image, MSE adaptation performs below the source model, whereas LPN-PM improves the source result from $11.52$ to $12.72$ dB.
The largest advantage occurs with five target images and decreases as more adaptation data become available.
This trend shows that proximal matching is most useful in the few-shot regime.

\noindent
The GS results show a stronger difference between the two adaptation objectives in the low-data regime. With one and five target images, GS-MSE drops below the source-model PSNR, while GS-AdaPM reaches
$24.98$ and $25.17$ dB. These results correspond to gains of
$12.69$ and $10.00$ dB over GS-MSE.
With 25 and 50 images, GS-AdaPM improves PSNR by
$2.10$ and $2.30$ dB, respectively.
It also produces lower relative trajectory mismatch at every adaptation budget.

\noindent
For one-shot and five-shot adaptation, the empirical Jacobian estimates of GS-AdaPM are much smaller than those of GS-MSE.
At larger adaptation budgets, both methods yield estimates close to one.
These results indicate that the contractivity-regularized proximal-matching objective better preserves the gradient-step structure when target data are limited.

\begin{figure}[htp!]
    \centering
    \includegraphics[width=1.0\linewidth]{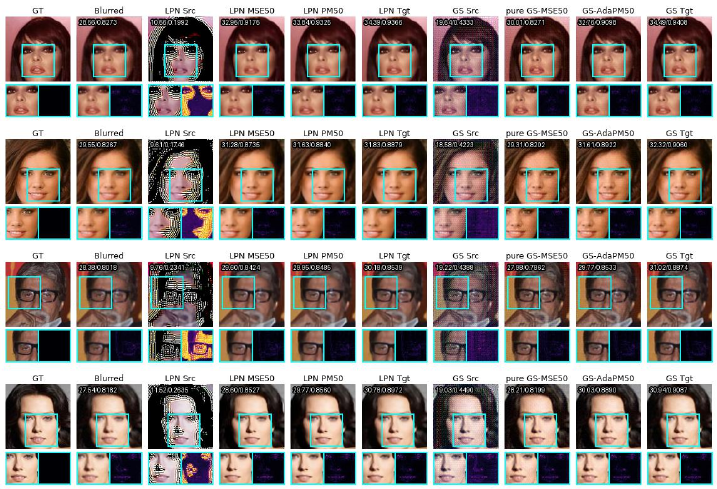}
    \caption{Representative Gaussian deblurring results for
    $n_{\mathrm{adapt}}=50$ under the BreCaHAD-to-CelebA domain shift.
    The source-trained denoisers produce visible domain-shift artifacts.
    LPN-PM and GS-AdaPM recover sharper facial contours and fewer local distortions than their MSE-adapted counterparts.
    The boxes mark the enlarged regions shown below each reconstruction.}
    \label{fig:deblurring-qualitative}
\end{figure}

\noindent
The reconstructions in
Figure~\ref{fig:deblurring-qualitative} follow the quantitative results.
The source-trained denoisers produce strong artifacts after deployment on CelebA.
Both adaptation objectives reduce these artifacts, but the proximal-matching variants recover sharper facial contours and retain more local detail than the corresponding MSE baselines.
Across both denoiser families, lower relative trajectory mismatch is associated with higher deblurring PSNR.

\subsection{Main Super-Resolution Results}

Table~\ref{tab:sr_lpn_gs_main} and
Figure~\ref{fig:SR-2x2} report the $\times4$ super-resolution results under the BreCaHAD-to-CelebA domain shift. The source-trained models show a large loss in reconstruction quality. LPN-source achieves $9.94$ dB, compared with $27.67$ dB for the target-trained LPN reference. GS-source achieves $18.08$ dB, compared with $27.67$ dB for the target-trained GS reference. Both source models also exhibit nonzero relative trajectory mismatch, showing that the reconstruction loss is accompanied by disagreement with the corresponding target proximal responses.

\begin{figure}[htp!]
    \centering
    \includegraphics[width=0.95\linewidth]{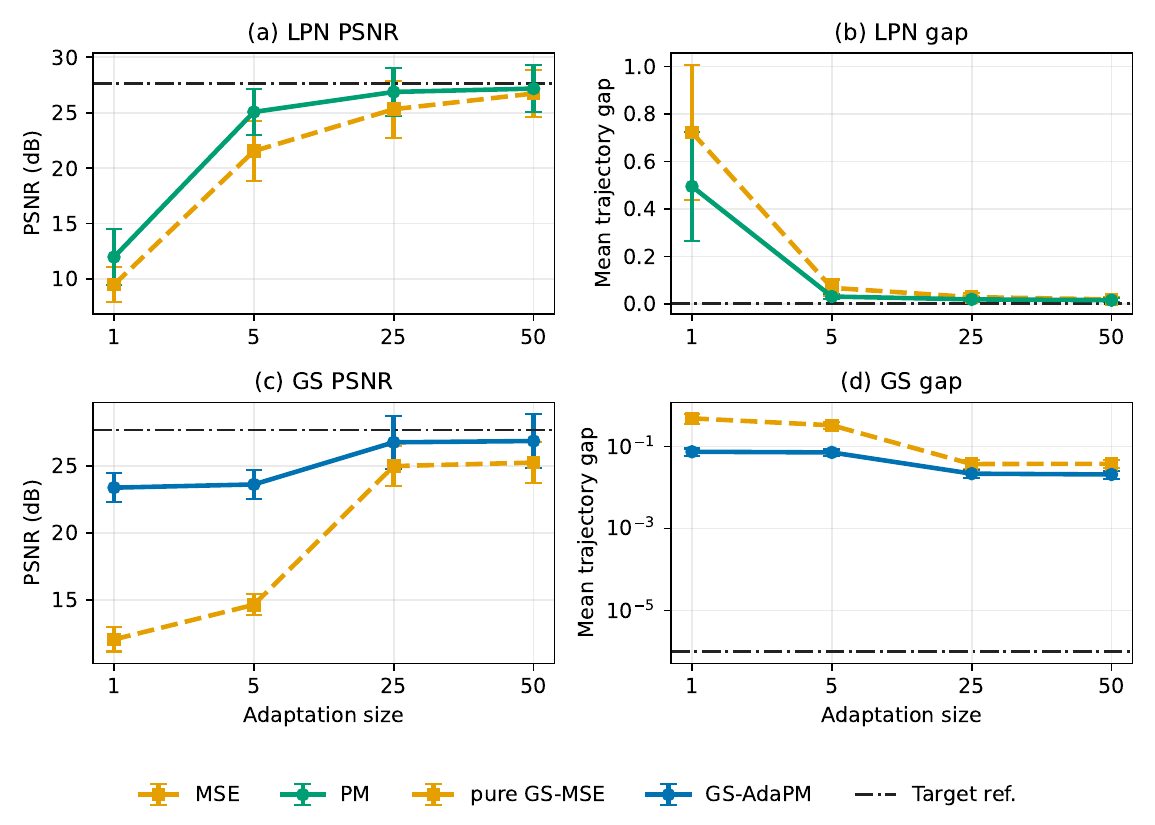}
    \caption{$\times4$ super-resolution performance under the BreCaHAD-to-CelebA domain shift.
    PSNR and relative trajectory mismatch are shown as functions of the number of target-domain adaptation images.
    Proximal-matching adaptation yields the largest PSNR gains in the low-data regime and reduces the mismatch with the corresponding target-trained reference.
    Higher PSNR and lower mismatch are better.}
    \label{fig:SR-2x2}
\end{figure}

\begin{table*}[t]
\centering
\caption{
Super-resolution results under severe BreCaHAD-to-CelebA proximal
mismatch. Values are mean $\pm$ standard deviation over 100 test
images. SSIM and trajectory gap (Gap) are reported in units of
$10^{-2}$. The gap is computed against the target-trained reference
within each denoiser family. $\widehat L_{\max}$ denotes the empirical
local Jacobian estimate for GS denoisers.
}
\label{tab:sr_lpn_gs_main}

\begingroup

\sisetup{
    text-series-to-math      = true,
    uncertainty-mode         = separate,
    table-align-uncertainty  = true,
    allow-uncertainty-breaks = false,
    table-alignment-mode     = format,
    table-number-alignment   = center,
    table-text-alignment     = center,
    table-model-setup        = \bfseries,
    group-digits             = none
}

\small
\setlength{\tabcolsep}{4.8pt}
\renewcommand{\arraystretch}{1.12}

\begin{tabular}{
    @{}
    l
    c
    S[table-format = 2.2(1.2)]
    S[table-format = 2.1(2.1)]
    S[table-format = 2.2(2.2)]
    S[table-format = 1.3]
    @{}
}
\toprule
\textbf{Method}
&
\textbf{$n_{\mathrm{adapt}}$}
&
{\textbf{PSNR $\uparrow$}}
&
{\textbf{SSIM $(10^{-2})$ $\uparrow$}}
&
{\textbf{Gap $(10^{-2})$ $\downarrow$}}
&
{\textbf{$\widehat L_{\max}$ $\downarrow$}}
\\
\midrule

\multicolumn{6}{@{}l}{\textit{LPN proximal denoiser}} \\
\cmidrule(lr){1-6}

LPN-source
& \multicolumn{1}{c}{--}
& 9.94(1.76)
& 23.7(14.3)
& 89.61(39.40)
& \multicolumn{1}{c}{--}
\\

\rowcolor{referencegray}
LPN-reference
& \multicolumn{1}{c}{--}
& 27.67(2.34)
& 81.2(5.8)
& \multicolumn{1}{c}{0}
& \multicolumn{1}{c}{--}
\\

\addlinespace[2pt]
LPN-MSE
&
& 9.48(1.56)
& 24.8(13.3)
& 72.28(28.46)
& \multicolumn{1}{c}{--}
\\

\rowcolor{methodgray}
\textbf{LPN-PM}
& \multirow{-2}{*}{$1$}
& \bfseries 11.97(2.56)
& \bfseries 38.5(15.4)
& \bfseries 49.53(23.08)
& \multicolumn{1}{c}{--}
\\

\addlinespace[2pt]
LPN-MSE
&
& 21.54(2.74)
& 65.1(8.4)
& 6.79(3.52)
& \multicolumn{1}{c}{--}
\\

\rowcolor{methodgray}
\textbf{LPN-PM}
& \multirow{-2}{*}{$5$}
& \bfseries 25.07(2.10)
& \bfseries 73.9(6.0)
& \bfseries 3.12(1.12)
& \multicolumn{1}{c}{--}
\\

\addlinespace[2pt]
LPN-MSE
&
& 25.31(2.56)
& 75.2(6.2)
& 2.88(1.75)
& \multicolumn{1}{c}{--}
\\

\rowcolor{methodgray}
\textbf{LPN-PM}
& \multirow{-2}{*}{$25$}
& \bfseries 26.87(2.20)
& \bfseries 78.8(5.5)
& \bfseries 1.93(0.98)
& \multicolumn{1}{c}{--}
\\

\addlinespace[2pt]
LPN-MSE
&
& 26.73(2.13)
& 78.0(5.5)
& 1.92(1.01)
& \multicolumn{1}{c}{--}
\\

\rowcolor{methodgray}
\textbf{LPN-PM}
& \multirow{-2}{*}{$50$}
& \bfseries 27.18(2.14)
& \bfseries 79.4(5.5)
& \bfseries 1.58(0.77)
& \multicolumn{1}{c}{--}
\\

\addlinespace[4pt]
\multicolumn{6}{@{}l}{\textit{GS proximal denoiser}} \\
\cmidrule(lr){1-6}

GS-source
& \multicolumn{1}{c}{--}
& 18.08(1.10)
& 37.9(8.2)
& 15.35(4.65)
& 1.043
\\

\rowcolor{referencegray}
GS-reference
& \multicolumn{1}{c}{--}
& 27.67(2.35)
& 82.3(5.1)
& \multicolumn{1}{c}{0}
& 0.927
\\

\addlinespace[2pt]
GS-MSE
&
& 12.04(0.90)
& 13.4(3.8)
& 48.82(12.38)
& 6.874
\\

\rowcolor{methodgray}
\textbf{GS-AdaPM}
& \multirow{-2}{*}{$1$}
& \bfseries 23.37(1.07)
& \bfseries 58.2(4.1)
& \bfseries 7.47(1.50)
& 1.438
\\

\addlinespace[2pt]
GS-MSE
&
& 14.62(0.79)
& 19.9(3.5)
& 33.03(6.65)
& 5.204
\\

\rowcolor{methodgray}
\textbf{GS-AdaPM}
& \multirow{-2}{*}{$5$}
& \bfseries 23.61(1.08)
& \bfseries 60.1(4.1)
& \bfseries 7.19(1.30)
& 1.002
\\

\addlinespace[2pt]
GS-MSE
&
& 24.98(1.48)
& 67.0(4.6)
& 3.74(0.84)
& 0.989
\\

\rowcolor{methodgray}
\textbf{GS-AdaPM}
& \multirow{-2}{*}{$25$}
& \bfseries 26.75(1.96)
& \bfseries 78.0(4.5)
& \bfseries 2.18(0.50)
& 0.985
\\

\addlinespace[2pt]
GS-MSE
&
& 25.24(1.56)
& 68.6(4.6)
& 3.80(0.82)
& 0.991
\\

\rowcolor{methodgray}
\textbf{GS-AdaPM}
& \multirow{-2}{*}{$50$}
& \bfseries 26.85(1.99)
& \bfseries 78.4(4.6)
& \bfseries 2.08(0.49)
& 0.985
\\

\bottomrule
\end{tabular}

\endgroup
\end{table*}

\noindent
For LPNs, proximal matching achieves higher PSNR, higher SSIM, and lower relative trajectory mismatch than MSE adaptation at every adaptation budget. The PSNR gains over LPN-MSE are
$2.49$, $3.53$, $1.56$, and $0.45$ dB for
$n_{\mathrm{adapt}}\in\{1,5,25,50\}$, respectively.
With one target image, LPN-MSE performs below the source model, whereas LPN-PM improves the source PSNR from $9.94$ to $11.97$ dB.
The largest gain occurs with five target images.
As the adaptation set grows, both methods approach the target-trained reference and their performance gap narrows.

\noindent
The GS results show a larger difference between the two adaptation objectives in the low-data regime. With one and five target images, GS-MSE falls below the source-model PSNR, while GS-AdaPM achieves
$23.37$ and $23.61$ dB.
These results correspond to gains of
$11.33$ and $8.99$ dB over GS-MSE.
With 25 and 50 target images, GS-AdaPM improves PSNR by
$1.77$ and $1.61$ dB, respectively.
It also produces lower relative trajectory mismatch at every adaptation budget.

\noindent
For one-shot and five-shot adaptation, the empirical Jacobian estimates of GS-AdaPM are much smaller than those of GS-MSE.
For 25 and 50 target images, both methods yield estimates close to one.
These results show that the contractivity-regularized objective better controls the empirical local Jacobian during few-shot adaptation.

\begin{figure}[htp!]
    \centering
    \includegraphics[width=1.0\linewidth]{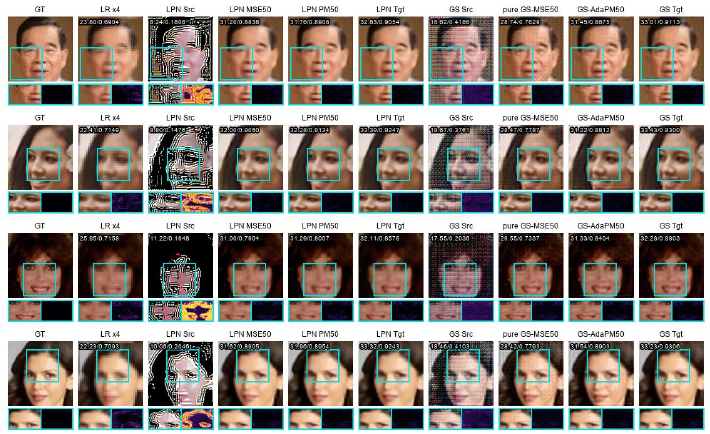}
    \caption{Representative $\times4$ super-resolution results for
    $n_{\mathrm{adapt}}=50$ under the BreCaHAD-to-CelebA domain shift.
    The proximal-matching variants recover sharper facial contours and fewer local distortions than the corresponding MSE-adapted models.
    The boxes mark the enlarged regions shown below each reconstruction.}
    \label{fig:SR-qualitative}
\end{figure}

\noindent
The visual comparisons in
Figure~\ref{fig:SR-qualitative} follow the quantitative results.
The source-trained models produce strong smoothing and visible artifacts around facial contours.
Both adaptation objectives improve the reconstructions, but the proximal-matching variants preserve sharper boundaries and more coherent local structure than the corresponding MSE baselines.

\noindent
Together with the deblurring results, these experiments show the same overall pattern across two forward operators.
Proximal-matching adaptation provides the largest gains when target-domain data are limited, and lower relative trajectory mismatch is associated with higher reconstruction quality.

\subsection{Further analysis}
Figure~\ref{fig:further} examines the effect of LPN adaptation along the PnP-PGD trajectory. At the same adaptation budget, PM-adapted denoisers generally achieve higher PSNR across iterations than their MSE-adapted counterparts. The difference is largest with one or five target images and decreases as the adaptation set grows. This trajectory-level trend agrees with the final reconstruction results in Tables~\ref{tab:deblurring_lpn_gs_main} and~\ref{tab:sr_lpn_gs_main}.

\noindent
The right-hand panels of Figure~\ref{fig:further} compare final PSNR with the average relative trajectory mismatch. Across both inverse problems, lower mismatch generally coincides with higher reconstruction PSNR. At each adaptation budget, proximal matching typically moves the result toward the upper-left region of the plot, corresponding to lower mismatch and higher PSNR. The same pattern appears for deblurring and super-resolution, supporting relative trajectory mismatch as a diagnostic of adaptation behavior.

\begin{figure}[htp!]
    \centering
    \includegraphics[width=1.0\linewidth]{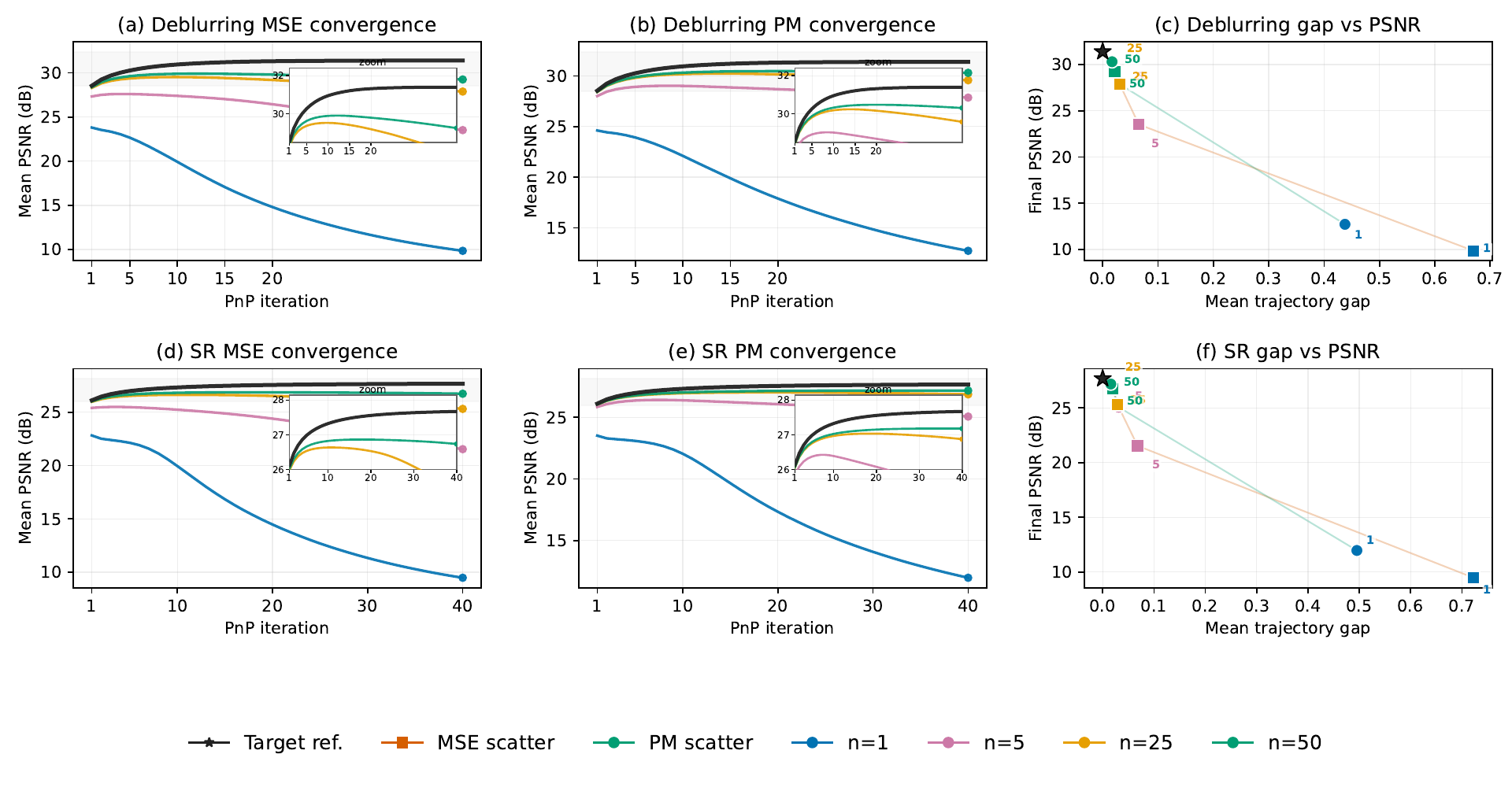}
    \caption{
    Trajectory-level analysis of LPN adaptation for Gaussian deblurring and
    $\times4$ super-resolution.
    The left and middle columns report mean PnP-PGD PSNR as a function of iteration for MSE- and PM-adapted denoisers, respectively.
    The right column plots final PSNR against the average relative trajectory mismatch.
    Proximal matching generally produces higher-PSNR trajectories and shifts the final results toward lower mismatch and higher reconstruction quality.
    Curves and points are averaged over 100 test images.
    }
    \label{fig:further}
\end{figure}

\noindent
Additional qualitative comparisons across adaptation budgets are provided in
Appendix~\ref{app:adaptation_budget}. For LPN-PM, the largest visual improvement occurs between one and five target images, followed by smaller refinements at larger budgets. GS-AdaPM changes less between one and five images and improves more clearly when the adaptation budget increases from five to 25 images. These observations agree with the quantitative trends in both inverse problems.

\section{Conclusion}
\label{sec:conclusion}

We studied domain shift in PnP-PGD when the deployed denoiser differs from a proximal map associated with the target domain. We formalized this discrepancy as \emph{proximal mismatch} and showed that it turns each deployed denoising update into an inexact proximal step for the target objective. Under regularity of the target proximal subproblems, we derived a nonasymptotic stationarity bound with an $\mathcal{O}(1/K)$ optimization term and an additive term proportional to the average squared proximal mismatch. The bound identifies persistent denoiser mismatch as a source of residual stationarity error.

\noindent
Guided by this result, we proposed few-shot adaptation based on proximal matching rather than MSE alone. We applied the same principle to Learned Proximal Networks and Gradient-Step denoisers, with contractivity regularization for the latter. Experiments on Gaussian deblurring and $\times4$ super-resolution under a BreCaHAD-to-CelebA domain shift showed that proximal matching improves reconstruction over MSE-based adaptation across both denoiser families. The gains were largest when only a few target-domain images were available and were accompanied by lower relative trajectory mismatch with the corresponding target-trained references.

\noindent
These results suggest that adapting a PnP denoiser should account for its role inside the reconstruction algorithm, rather than treating it only as a stand-alone denoising model. Our analysis assumes that the target prior admits a structured proximal reference, while the experiments use a target-trained reference to evaluate mismatch. Extending proximal matching to settings without such a reference, incorporating query points generated by the reconstruction trajectory, and evaluating broader domain shifts and inverse problems are important directions for future work.

\clearpage

\appendix

\section{Proofs}
\label{app:proofs}

\subsection{Proof of Lemma~\ref{lemma:bridge_discrepancy}}
\label{app:proof-bridge}

\begin{proof}
Let
\[
    \widehat{\boldsymbol{u}}_k
    =
    \widehat{\mathsf{D}}(\boldsymbol{z}_k),
    \qquad
    \boldsymbol{u}_k^\star
    =
    \mathsf{D}_\star(\boldsymbol{z}_k).
\]
By \(L_H\)-smoothness of \(H_k\) on the segment between
\(\widehat{\boldsymbol{u}}_k\) and \(\boldsymbol{u}_k^\star\),
\[
H_k(\widehat{\boldsymbol{u}}_k)
\leq
H_k(\boldsymbol{u}_k^\star)
+
\left\langle
    \nabla H_k(\boldsymbol{u}_k^\star),
    \widehat{\boldsymbol{u}}_k-\boldsymbol{u}_k^\star
\right\rangle
+
\frac{L_H}{2}
\|\widehat{\boldsymbol{u}}_k-\boldsymbol{u}_k^\star\|^2 .
\]
Since \(\nabla H_k(\boldsymbol{u}_k^\star)=0\), this gives
\[
    \epsilon_k
    =
    H_k(\widehat{\boldsymbol{u}}_k)
    -
    H_k(\boldsymbol{u}_k^\star)
    \leq
    \frac{L_H}{2}
    \|\widehat{\boldsymbol{u}}_k-\boldsymbol{u}_k^\star\|^2
    =
    \frac{L_H}{2}d_k^2 ,
\]
which is~\eqref{eq:bridge_smooth}.
\end{proof}

\subsection{Proof of Lemma~\ref{lemma:abstract_descent}}
\label{app:proof-descent}

\begin{proof}
By definition of the proximal objective error,
\[
    \epsilon_{k+1}
    =
    H_{k+1}(\boldsymbol{x}_{k+1})
    -
    H_{k+1}(\mathsf{D}_\star(\boldsymbol{z}_{k+1})).
\]
Since \(\mathsf{D}_\star(\boldsymbol{z}_{k+1})\) minimizes \(H_{k+1}\),
\[
    H_{k+1}(\mathsf{D}_\star(\boldsymbol{z}_{k+1}))
    \leq
    H_{k+1}(\boldsymbol{x}_k).
\]
Therefore,
\[
    H_{k+1}(\boldsymbol{x}_{k+1})
    \leq
    H_{k+1}(\boldsymbol{x}_k)+\epsilon_{k+1}.
\]
Expanding \(H_{k+1}\) yields
\[
R_\star(\boldsymbol{x}_{k+1})
+
\frac{1}{2}
\|\boldsymbol{x}_{k+1}-\boldsymbol{z}_{k+1}\|^2
\leq
R_\star(\boldsymbol{x}_k)
+
\frac{1}{2}
\|\boldsymbol{x}_k-\boldsymbol{z}_{k+1}\|^2
+
\epsilon_{k+1}.
\]
Using
\[
    \boldsymbol{z}_{k+1}
    =
    \boldsymbol{x}_k-\eta\nabla f(\boldsymbol{x}_k),
\]
we have
\[
    \boldsymbol{x}_k-\boldsymbol{z}_{k+1}
    =
    \eta\nabla f(\boldsymbol{x}_k),
\]
and
\[
    \boldsymbol{x}_{k+1}-\boldsymbol{z}_{k+1}
    =
    \boldsymbol{x}_{k+1}-\boldsymbol{x}_k
    +
    \eta\nabla f(\boldsymbol{x}_k).
\]
Hence
\[
\frac{1}{2}\|\boldsymbol{x}_k-\boldsymbol{z}_{k+1}\|^2
-
\frac{1}{2}\|\boldsymbol{x}_{k+1}-\boldsymbol{z}_{k+1}\|^2
=
-
\eta
\left\langle
    \boldsymbol{x}_{k+1}-\boldsymbol{x}_k,
    \nabla f(\boldsymbol{x}_k)
\right\rangle
-
\frac{1}{2}
\|\boldsymbol{x}_{k+1}-\boldsymbol{x}_k\|^2 .
\]
Thus
\begin{equation}
\begin{split}
R_\star(\boldsymbol{x}_{k+1})
&\leq
R_\star(\boldsymbol{x}_k)
-
\eta
\left\langle
    \boldsymbol{x}_{k+1}-\boldsymbol{x}_k,
    \nabla f(\boldsymbol{x}_k)
\right\rangle
\\
&\quad
-
\frac{1}{2}
\|\boldsymbol{x}_{k+1}-\boldsymbol{x}_k\|^2
+
\epsilon_{k+1}.
\end{split}
\label{eq:R_descent_expanded}
\end{equation}
By \(L_f\)-smoothness of \(f\),
\begin{equation}
\eta f(\boldsymbol{x}_{k+1})
\leq
\eta f(\boldsymbol{x}_k)
+
\eta
\left\langle
    \boldsymbol{x}_{k+1}-\boldsymbol{x}_k,
    \nabla f(\boldsymbol{x}_k)
\right\rangle
+
\frac{\eta L_f}{2}
\|\boldsymbol{x}_{k+1}-\boldsymbol{x}_k\|^2 .
\label{eq:f_descent_expanded}
\end{equation}
Adding~\eqref{eq:R_descent_expanded} and~\eqref{eq:f_descent_expanded} gives
\[
    F_\star(\boldsymbol{x}_{k+1})
    \leq
    F_\star(\boldsymbol{x}_k)
    -
    \frac{1-\eta L_f}{2}
    \|\boldsymbol{x}_{k+1}-\boldsymbol{x}_k\|^2
    +
    \epsilon_{k+1},
\]
which proves~\eqref{eq:abstract_descent}.
\end{proof}

\subsection{Proof of Lemma~\ref{lemma:abstract_residual}}
\label{app:proof-residual}

\begin{proof}
Let
\[
    \boldsymbol{r}_{k+1}
    :=
    -\nabla H_{k+1}(\boldsymbol{x}_{k+1}).
\]
Since
\[
    H_{k+1}(\boldsymbol{u})
    =
    \frac{1}{2}
    \|\boldsymbol{u}-\boldsymbol{z}_{k+1}\|^2
    +
    R_\star(\boldsymbol{u}),
\]
we have
\[
    \nabla H_{k+1}(\boldsymbol{x}_{k+1})
    =
    \boldsymbol{x}_{k+1}-\boldsymbol{z}_{k+1}
    +
    \nabla R_\star(\boldsymbol{x}_{k+1}).
\]
Therefore,
\[
    \nabla R_\star(\boldsymbol{x}_{k+1})
    =
    \boldsymbol{z}_{k+1}
    -
    \boldsymbol{x}_{k+1}
    -
    \boldsymbol{r}_{k+1}.
\]
This proves the residual identity in~\eqref{eq:abstract_residual}.
\noindent
For the norm bound, Assumption~\ref{assump:prox_subproblem_regularity} gives
\[
\|\nabla H_{k+1}(\boldsymbol{x}_{k+1})\|^2
\leq
2L_H
\left(
    H_{k+1}(\boldsymbol{x}_{k+1})
    -
    H_{k+1}(\mathsf{D}_\star(\boldsymbol{z}_{k+1}))
\right)
=
2L_H\epsilon_{k+1}.
\]
Since
\[
    \|\boldsymbol{r}_{k+1}\|
    =
    \|\nabla H_{k+1}(\boldsymbol{x}_{k+1})\|,
\]
the result follows.
\end{proof}

\subsection{Proof of Theorem~\ref{thm:main_stationarity}}
\label{app:proof-main}

\begin{proof}
By Lemma~\ref{lemma:abstract_residual},
\[
\begin{aligned}
\nabla F_\star(\boldsymbol{x}_{k+1})
&=
\eta\nabla f(\boldsymbol{x}_{k+1})
+
\nabla R_\star(\boldsymbol{x}_{k+1})
\\
&=
\boldsymbol{x}_k-\boldsymbol{x}_{k+1}
+
\eta
\big(
    \nabla f(\boldsymbol{x}_{k+1})
    -
    \nabla f(\boldsymbol{x}_k)
\big)
-
\boldsymbol{r}_{k+1}.
\end{aligned}
\]
Using the Lipschitz continuity of \(\nabla f\),
\[
\left\|
    \boldsymbol{x}_k-\boldsymbol{x}_{k+1}
    +
    \eta
    \big(
        \nabla f(\boldsymbol{x}_{k+1})
        -
        \nabla f(\boldsymbol{x}_k)
    \big)
\right\|
\leq
(1+\eta L_f)
\|\boldsymbol{x}_{k+1}-\boldsymbol{x}_k\|.
\]
Thus,
\[
\|\nabla F_\star(\boldsymbol{x}_{k+1})\|^2
\leq
2(1+\eta L_f)^2
\|\boldsymbol{x}_{k+1}-\boldsymbol{x}_k\|^2
+
2\|\boldsymbol{r}_{k+1}\|^2.
\]
By Lemma~\ref{lemma:abstract_residual},
\[
    2\|\boldsymbol{r}_{k+1}\|^2
    \leq
    4L_H\epsilon_{k+1}.
\]
By Lemma~\ref{lemma:abstract_descent},
\[
\|\boldsymbol{x}_{k+1}-\boldsymbol{x}_k\|^2
\leq
\frac{2}{1-\eta L_f}
\left(
    F_\star(\boldsymbol{x}_k)
    -
    F_\star(\boldsymbol{x}_{k+1})
    +
    \epsilon_{k+1}
\right).
\]
Combining the last three displays gives
\[
\begin{aligned}
\|\nabla F_\star(\boldsymbol{x}_{k+1})\|^2
&\leq
\frac{4(1+\eta L_f)^2}{1-\eta L_f}
\left(
    F_\star(\boldsymbol{x}_k)
    -
    F_\star(\boldsymbol{x}_{k+1})
\right)
\\
&\quad+
\left(
    \frac{4(1+\eta L_f)^2}{1-\eta L_f}
    +
    4L_H
\right)
\epsilon_{k+1}.
\end{aligned}
\]
Since \(\eta L_f<1\), \((1+\eta L_f)^2\leq4\), and hence
\[
\|\nabla F_\star(\boldsymbol{x}_{k+1})\|^2
\leq
\frac{16}{1-\eta L_f}
\left(
    F_\star(\boldsymbol{x}_k)
    -
    F_\star(\boldsymbol{x}_{k+1})
\right)
+
\left(
    \frac{16}{1-\eta L_f}
    +
    4L_H
\right)
\epsilon_{k+1}.
\]
Summing over \(k=0,\ldots,t-1\) yields
\[
\begin{aligned}
\sum_{k=0}^{t-1}
\|\nabla F_\star(\boldsymbol{x}_{k+1})\|^2
&\leq
\frac{16}{1-\eta L_f}
\sum_{k=0}^{t-1}
\left(
    F_\star(\boldsymbol{x}_k)
    -
    F_\star(\boldsymbol{x}_{k+1})
\right)
\\
&\quad+
\left(
    \frac{16}{1-\eta L_f}
    +
    4L_H
\right)
\sum_{k=0}^{t-1}
\epsilon_{k+1}.
\end{aligned}
\]
The objective terms telescope, and Assumption~\ref{assump:lower_bound} gives
\[
\sum_{k=0}^{t-1}
\left(
    F_\star(\boldsymbol{x}_k)
    -
    F_\star(\boldsymbol{x}_{k+1})
\right)
=
F_\star(\boldsymbol{x}_0)-F_\star(\boldsymbol{x}_t)
\leq
F_\star(\boldsymbol{x}_0)-F_\star^{\inf}.
\]
After reindexing,
\begin{equation}
\sum_{k=1}^{t}
\|\nabla F_\star(\boldsymbol{x}_k)\|^2
\leq
\frac{16}{1-\eta L_f}
\big(
    F_\star(\boldsymbol{x}_0)-F_\star^{\inf}
\big)
+
\left(
    \frac{16}{1-\eta L_f}
    +
    4L_H
\right)
\sum_{k=1}^{t}\epsilon_k .
\label{eq:epsilon_parameterized_rate}
\end{equation}
Lemma~\ref{lemma:bridge_discrepancy} gives
\[
    \epsilon_k
    \leq
    \frac{L_H}{2}d_k^2.
\]
Substituting this into~\eqref{eq:epsilon_parameterized_rate} and dividing by \(t\) proves~\eqref{eq:abstract_rate}. The bound on the minimum follows from
\[
    \min_{1\leq k\leq t}a_k
    \leq
    \frac{1}{t}\sum_{k=1}^{t}a_k
\]
for nonnegative \(a_k\).
\noindent
Finally, if \(\sum_{k=1}^{\infty}d_k^2<\infty\), then
\[
    \sum_{k=1}^{\infty}\epsilon_k<\infty .
\]
Letting \(t\to\infty\) in~\eqref{eq:epsilon_parameterized_rate} gives
\[
    \sum_{k=1}^{\infty}
    \|\nabla F_\star(\boldsymbol{x}_k)\|^2
    <
    \infty .
\]
Therefore
\[
    \|\nabla F_\star(\boldsymbol{x}_k)\|\to0 .
\]
\end{proof}

\subsection{Proof of Corollary~\ref{cor:exact_target}}
\label{app:proof-exact-target}

\begin{proof}
If \(\widehat{\mathsf{D}}=\mathsf{D}_\star\), then
\[
    d_k
    =
    \|\widehat{\mathsf{D}}(\boldsymbol{z}_k)
    -
    \mathsf{D}_\star(\boldsymbol{z}_k)\|
    =
    0
\]
for all visited query points. Lemma~\ref{lemma:bridge_discrepancy} gives
\[
    \epsilon_k
    \leq
    \frac{L_H}{2}d_k^2
    =
    0.
\]
Since \(\epsilon_k\geq0\), we have \(\epsilon_k=0\). Substituting \(d_k=0\) into Theorem~\ref{thm:main_stationarity} gives
\[
    \frac{1}{t}\sum_{k=1}^{t}
    \|\nabla F_\star(\boldsymbol{x}_k)\|^2
    \leq
    \frac{16}{(1-\eta L_f)t}
    \big(
        F_\star(\boldsymbol{x}_0)-F_\star^{\inf}
    \big),
\]
which is~\eqref{eq:exact_rate}.
\end{proof}

\subsection{Proof of Proposition~\ref{prop:pm_vs_mse}}
\label{proof:prop_pm_vs_mse}

\begin{proof}
For a fixed \(y\), the conditional MSE risk decomposes as
\begin{align*}
    \mathbb E
    \left[
        \|\mathsf{D}(Y)-X\|_2^2
        \,\middle|\,
        Y=y
    \right]
    &=
    \left\|
        \mathsf{D}(y)-\mathbb E[X\mid Y=y]
    \right\|_2^2
    \\
    &\quad+
    \mathbb E
    \left[
        \left\|
            X-\mathbb E[X\mid Y=y]
        \right\|_2^2
        \,\middle|\,
        Y=y
    \right].
\end{align*}
The second term is independent of \(\mathsf{D}(y)\). Hence,
\[
    \mathsf{D}_{\mathrm{MSE}}^\star(y)
    =
    \mathbb E[X\mid Y=y].
\]

The conditional PM risk equals
\[
    1-
    \mathbb E
    \left[
        \exp
        \left(
            -\frac{\|\mathsf{D}(y)-X\|_2^2}{2\gamma^2}
        \right)
        \,\middle|\,
        Y=y
    \right].
\]
Its minimizers therefore maximize
\[
    z
    \mapsto
    \int
    \exp
    \left(
        -\frac{\|z-x\|_2^2}{2\gamma^2}
    \right)
    p_{\mathrm{tar}}(x\mid y)\,dx.
\]
Up to normalization, this is the Gaussian smoothing of
\(p_{\mathrm{tar}}(\cdot\mid y)\). The consistency result
of~\citet{fang2024whats} gives
\[
    \mathsf{D}_{\mathrm{PM},\gamma}^\star(y)
    \longrightarrow
    \arg\max_x p_{\mathrm{tar}}(x\mid y)
\]
as \(\gamma\downarrow0\), under the stated conditions.

Finally,
\[
    p_{\mathrm{tar}}(x\mid y)
    \propto
    \exp
    \left(
        -\frac{\|y-x\|_2^2}{2\sigma_{\mathrm{adapt}}^2}
    \right)
    p_{\mathrm{tar}}(x).
\]
Using \(R_{\mathrm{tar}}(x)=-\log p_{\mathrm{tar}}(x)\),
\[
    \arg\max_x p_{\mathrm{tar}}(x\mid y)
    =
    \operatorname{prox}_{\sigma_{\mathrm{adapt}}^2R_{\mathrm{tar}}}(y).
\]
\end{proof}

\section{Proximal Characterizations of Structured Denoisers}
\label{app:structured_proximal_realizations}

This appendix collects the proximal characterizations used in
Section~\ref{subsec:proximal_realizations} and derives the corresponding
smoothness constants for the proximal subproblems.

\subsection{Learned Proximal Networks}

Recall the LPN parameterization
\[
    \mathsf{D}_\theta=\nabla\Psi_\theta,
    \qquad
    \Psi_\theta(\boldsymbol{x})
    =
    \psi_\theta(\boldsymbol{x})
    +
    \frac{\alpha}{2}\|\boldsymbol{x}\|^2,
\]
where \(\psi_\theta\) is input-convex and \(\alpha\in(0,1)\).

\begin{myprop}[\citep{fang2024whats}]
\label{prop:lpn}
Let \(\mathsf{D}_\theta=\nabla\Psi_\theta\) be an LPN as
in~\eqref{eq:LPN_def}. Then:
\begin{itemize}
    \item[(i)] \(\mathsf{D}_\theta:\mathbb{R}^n\to\mathbb{R}^n\) is a
    continuous bijection with continuous inverse
    \[
        \mathsf{D}_\theta^{-1}
        =
        \nabla\Psi_\theta^\star,
    \]
    and \(\mathsf{D}_\theta^{-1}\) is \(1/\alpha\)-Lipschitz;

    \item[(ii)] there exists a real-valued, \(\mathcal{C}^1\), coercive
    function \(R_\theta:\mathbb{R}^n\to\mathbb{R}\) such that
    \[
        \mathsf{D}_\theta
        =
        \operatorname{prox}_{R_\theta};
    \]

    \item[(iii)] for every \(\boldsymbol{x}\in\mathbb{R}^n\),
    \[
        \nabla R_\theta(\boldsymbol{x})
        =
        \mathsf{D}_\theta^{-1}(\boldsymbol{x})
        -
        \boldsymbol{x}.
    \]
\end{itemize}
\end{myprop}

\paragraph{Proximal-subproblem regularity.}

Let
\[
    \mathsf{D}_\star=\mathsf{D}_{\theta^\star},
    \qquad
    R_\star=R_{\theta^\star}.
\]
For the proximal subproblem
\[
    H_k(\boldsymbol{u})
    =
    R_\star(\boldsymbol{u})
    +
    \frac{1}{2}
    \|\boldsymbol{u}-\boldsymbol{z}_k\|^2,
\]
Proposition~\ref{prop:lpn} gives
\begin{align*}
    \nabla H_k(\boldsymbol{u})
    &=
    \nabla R_\star(\boldsymbol{u})
    +
    \boldsymbol{u}
    -
    \boldsymbol{z}_k
    \\
    &=
    \mathsf{D}_{\theta^\star}^{-1}(\boldsymbol{u})
    -
    \boldsymbol{z}_k.
\end{align*}
Since
\[
    \mathsf{D}_{\theta^\star}^{-1}
    =
    \nabla\Psi_{\theta^\star}^\star,
\]
integration yields
\[
    H_k(\boldsymbol{u})
    =
    \Psi_{\theta^\star}^\star(\boldsymbol{u})
    -
    \langle\boldsymbol{z}_k,\boldsymbol{u}\rangle.
\]
The function \(H_k\) is convex, and its gradient is \(1/\alpha\)-Lipschitz.
Thus Assumption~\ref{assump:prox_subproblem_regularity} holds with
\[
    L_H=\frac{1}{\alpha}.
\]
Applying Lemma~\ref{lemma:bridge_discrepancy} with
\[
    \boldsymbol{u}_k
    =
    \widehat{\mathsf{D}}_\theta(\boldsymbol{z}_k),
    \qquad
    \boldsymbol{u}_k^\star
    =
    \mathsf{D}_{\theta^\star}(\boldsymbol{z}_k),
\]
gives
\[
    \epsilon_k
    \leq
    \frac{1}{2\alpha}
    \|\boldsymbol{u}_k-\boldsymbol{u}_k^\star\|^2,
\]
which is precisely~\eqref{eq:lpn_eps}.

\subsection{Gradient-Step Proximal Denoisers}

Recall that a GS denoiser is defined by
\[
    \mathsf{D}_\sigma
    =
    \operatorname{Id}-\nabla g_\sigma
    =
    \nabla h_\sigma,
    \qquad
    h_\sigma(\boldsymbol{x})
    =
    \frac12\|\boldsymbol{x}\|^2-g_\sigma(\boldsymbol{x}).
\]

\begin{myprop}[GS proximal interpretation; \citep{hurault2022proximal}]
\label{prop:gs_prox}
Let \(\mathcal{X}\subset\mathbb{R}^n\) be open and convex, and let
\(g_\sigma:\mathcal{X}\to\mathbb{R}\) be sufficiently smooth. If
\(\nabla g_\sigma\) is \(L\)-Lipschitz with \(L<1\), then:
\begin{itemize}
    \item[(i)] \(h_\sigma\) is \((1-L)\)-strongly convex;

    \item[(ii)] \(\mathsf{D}_\sigma\) is injective;

    \item[(iii)] there exists an induced regularizer \(\phi_\sigma\) such
    that
    \[
        \mathsf{D}_\sigma
        =
        \operatorname{prox}_{\phi_\sigma}
        \qquad
        \text{on }
        \operatorname{Im}(\mathsf{D}_\sigma);
    \]

    \item[(iv)] for every
    \(\boldsymbol{x}\in\operatorname{Im}(\mathsf{D}_\sigma)\),
    \[
        \nabla\phi_\sigma(\boldsymbol{x})
        =
        \mathsf{D}_\sigma^{-1}(\boldsymbol{x})
        -
        \boldsymbol{x};
    \]

    \item[(v)] \(\mathsf{D}_\sigma^{-1}\) is
    \(1/(1-L)\)-Lipschitz.
\end{itemize}
\end{myprop}

\paragraph{Proximal-subproblem regularity.}

Let
\[
    \mathsf{D}_\star=\mathsf{D}_{\sigma,\star},
    \qquad
    R_\star=\phi_{\sigma,\star}.
\]
Because the induced regularizer is defined through the image of the target
denoiser, assume that
\[
    \widehat{\mathsf{D}}_\sigma(\boldsymbol{z}_k)
    \in
    \operatorname{Im}(\mathsf{D}_{\sigma,\star})
\]
for every visited query point. This ensures that
\(\phi_{\sigma,\star}\) is finite and differentiable at each produced
iterate.

\noindent
For
\[
    H_k(\boldsymbol{u})
    =
    \phi_{\sigma,\star}(\boldsymbol{u})
    +
    \frac12
    \|\boldsymbol{u}-\boldsymbol{z}_k\|^2,
\]
we have, on \(\operatorname{Im}(\mathsf{D}_{\sigma,\star})\),
\begin{align*}
    \nabla H_k(\boldsymbol{u})
    &=
    \nabla\phi_{\sigma,\star}(\boldsymbol{u})
    +
    \boldsymbol{u}
    -
    \boldsymbol{z}_k
    \\
    &=
    \mathsf{D}_{\sigma,\star}^{-1}(\boldsymbol{u})
    -
    \boldsymbol{z}_k.
\end{align*}
Since
\[
    \mathsf{D}_{\sigma,\star}^{-1}
    =
    \nabla h_{\sigma,\star}^\star,
\]
it follows, up to an additive constant, that
\[
    H_k(\boldsymbol{u})
    =
    h_{\sigma,\star}^\star(\boldsymbol{u})
    -
    \langle\boldsymbol{z}_k,\boldsymbol{u}\rangle.
\]
Because \(h_{\sigma,\star}\) is \((1-L)\)-strongly convex, its conjugate has
a \(1/(1-L)\)-Lipschitz gradient. Hence \(H_k\) is convex and
\(1/(1-L)\)-smooth, and
\[
    L_H=\frac{1}{1-L}.
\]
Applying Lemma~\ref{lemma:bridge_discrepancy} with
\[
    \boldsymbol{u}_k
    =
    \widehat{\mathsf{D}}_\sigma(\boldsymbol{z}_k),
    \qquad
    \boldsymbol{u}_k^\star
    =
    \mathsf{D}_{\sigma,\star}(\boldsymbol{z}_k),
\]
gives
\[
    \epsilon_k
    \leq
    \frac{1}{2(1-L)}
    \|\boldsymbol{u}_k-\boldsymbol{u}_k^\star\|^2,
\]
which is~\eqref{eq:gs_eps}.

\section{Additional Experimental Results}
\label{app:additional_experiments}

\subsection{Effect of the Adaptation Budget}
\label{app:adaptation_budget}

This section provides qualitative comparisons across different numbers of target-domain adaptation images. The corresponding quantitative results are reported in Tables~\ref{tab:deblurring_lpn_gs_main} and~\ref{tab:sr_lpn_gs_main}. The figures show that the effect of the adaptation budget differs between the two denoiser families. LPN-PM improves most strongly between one and five target images, whereas GS-AdaPM shows a larger change between five and 25 images.

\begin{figure}[htp!]
    \centering
    \includegraphics[width=0.95\linewidth]{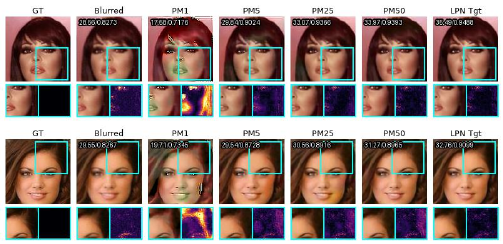}
    \caption{
    Representative Gaussian deblurring results for LPN-PM at different target-domain adaptation budgets.
    The largest visual improvement occurs between one and five adaptation images.
    Additional target images further reduce artifacts and refine facial contours and local details.
    }
    \label{fig:db-lpn-adaptsize}
\end{figure}

\begin{figure}[htp!]
    \centering
    \includegraphics[width=0.95\linewidth]{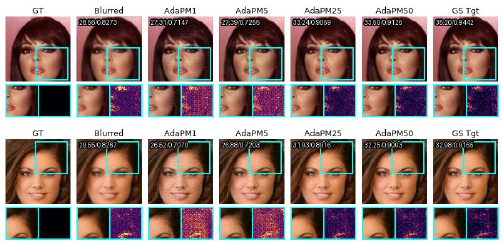}
    \caption{
    Representative Gaussian deblurring results for GS-AdaPM at different target-domain adaptation budgets.
    The reconstructions change less between one and five images and improve more clearly when 25 target images are used.
    Larger adaptation sets recover sharper facial contours and finer local structure.
    }
    \label{fig:db-gs-adaptsize}
\end{figure}

\begin{figure}[htp!]
    \centering
    \includegraphics[width=0.95\linewidth]{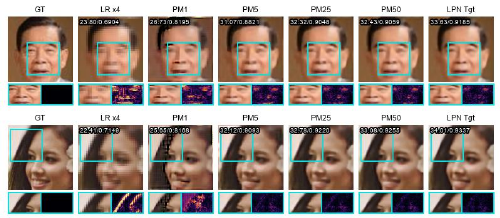}
    \caption{
    Representative $\times4$ super-resolution results for LPN-PM at different target-domain adaptation budgets.
    The largest visual improvement occurs between one and five target images.
    Further adaptation produces smaller refinements in facial boundaries and local detail.
    }
    \label{fig:SR-lpn-adaptsize}
\end{figure}

\begin{figure}[htp!]
    \centering
    \includegraphics[width=0.95\linewidth]{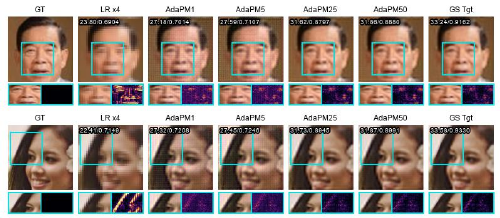}
    \caption{
    Representative $\times4$ super-resolution results for GS-AdaPM at different target-domain adaptation budgets.
    The main visual improvement occurs when the adaptation budget increases from five to 25 images.
    Larger budgets produce more coherent facial boundaries and fewer local distortions.
    }
    \label{fig:SR-gs-adaptsize}
\end{figure}

\clearpage

\bibliography{main.bib}

\end{document}